\documentclass[11pt]{article}

\usepackage{amssymb,amsmath,amsfonts}
\usepackage{amsthm}	
\usepackage{graphicx,color,enumitem,calrsfs,bm,verbatim}
\usepackage[colorlinks,pagebackref=true]{hyperref}
\usepackage{tikz-cd,pst-node}
\usepackage[round]{natbib}
\usepackage{geometry}

\newcommand{\e}{{\rm e}}

\newcommand{\E}{{\mathbb E}}

\renewcommand{\P}{{\mathbb P}}
\newcommand{\Q}{{\mathbb Q}}

\newcommand{\R}{{\mathbb R}}
\renewcommand{\S}{{\mathbb S}}
\newcommand{\N}{{\mathbb N}}

\newcommand{\Acal}{{\mathcal A}}
\newcommand{\Ecal}{{\mathcal E}}
\newcommand{\Fcal}{{\mathcal F}}
\newcommand{\Gcal}{{\mathcal G}}
\newcommand{\id}{{\rm id}}
\newcommand{\Pol}{{\rm Pol}}
\DeclareMathOperator{\vspan}{span}
\DeclareMathOperator{\tr}{Tr}
\newtheorem{thm}{Theorem}[section]
\newtheorem{lem}[thm]{Lemma}
\newtheorem{cor}[thm]{Corollary}
\newtheorem{rem}[thm]{Remark}
\newtheorem{ex}[thm]{Example}
\newtheorem{defi}[thm]{Definition}

\numberwithin{equation}{section}

\begin{document}

\title{Polynomial Jump-Diffusion Models\footnote{We thank Agostino Capponi, Elise Gourier, Jan Kallsen, Sergio Pulido, and the participants at the 16th Winter School on Mathematical Finance in Lunteren, School and Workshop on Dynamical Models in Finance at EPFL, 61st World Statistics Congress of the International Statistics Institute, and CEAR/Huebner Summer Risk Institute at Georgia State University, the associate editor and two anonymous referees for their comments. The research leading to these results has received funding from the European Research Council under the European Union's Seventh Framework Programme (FP/2007-2013) / ERC Grant Agreement n.~307465-POLYTE.}}
\author{Damir Filipovi\'c\thanks{EPFL and Swiss Finance Institute, Extranef 218, 1015 Lausanne, Switzerland, email: damir.filipovic@epfl.ch} \quad\quad  Martin Larsson\thanks{ETH Zurich, Department of Mathematics, R\"amistrasse 101, 8092 Zurich, Switzerland, email: martin.larsson@math.ethz.ch}}

\date{July 21, 2019\\[2ex] forthcoming in Stochastic Systems}

\maketitle

\begin{abstract}
We develop a comprehensive mathematical framework for polynomial jump-diffusions in a semimartingale context, which nest affine jump-diffusions and have broad applications in finance. We show that the polynomial property is preserved under polynomial transformations and L\'evy time change. We present a generic method for option pricing based on moment expansions. As an application, we introduce a large class of novel financial asset pricing models with excess log returns that are conditional L\'evy based on polynomial jump-diffusions.
\\[2ex]
\\[2ex]
\\[2ex]
\end{abstract}

\section{Introduction}

Polynomial jump-diffusions have broad applications in finance. A jump-diffusion is polynomial if its extended generator maps any polynomial to a polynomial of equal or lower degree. As a consequence, its conditional moments can be computed in closed form. This property renders polynomial jump-diffusions computationally tractable and perfectly suited for financial asset pricing models. Many commonly occurring jump-diffusions are polynomial, for example Ornstein--Uhlenbeck processes, square-root diffusions, Jacobi or Wright--Fisher diffusions, L\'evy processes and geometric L\'evy processes, as well as multi-dimensional analogues and combinations of these processes.

In this paper, we develop a comprehensive mathematical framework for polynomial jump-diffusions. We characterize the polynomial property in terms of the coefficients of the extended generator, and we establish the moment formula. We show that the polynomial property of a jump-diffusion is preserved under polynomial transformations and L\'evy time change. These transformations allow us to easily construct polynomial jump-diffusions from simple building blocks. They also allow us to efficiently specify non-linearities in financial models, which renders them flexible in capturing many of the empirical features of financial time series.

We also revisit affine jump-diffusions, which have been widely used in financial asset pricing models for the last two decades. We provide a relaxed definition of an affine jump-diffusion in terms of the pointwise action of its extended generator on exponential-affine functions, and we establish the affine transform formula. We show that, modulo integrability conditions on the jumps, affine jump-diffusions are polynomial. The converse does not hold, so that polynomial jump-diffusions truly extend the class of affine jump-diffusions. We find that the affine property of a jump-diffusion is not invariant under polynomial transformations or L\'evy time change in general, which may explain why these transformations have not been widely applied in affine financial models.

In contrast to earlier work on polynomial and affine processes, we do not require the Markov property. Instead, we work in a special semimartingale framework, which is more flexible and amenable to practical applications. Non-Markovian polynomial jump-diffusions exist and can be constructed using the counterexamples of \citet[Section~3]{kal_kru_16}, and most of our results apply to them. Indeed, Markovianity is only assumed for the invariance of the polynomial property under L\'evy time change, all other arguments rely on It\^o calculus rather than functional analysis.

We present a generic method for option pricing in polynomial jump-diffusion models. This method builds on the expansion of the likelihood ratio function with respect to an orthonormal basis of polynomials in some conveniently weighted $L^2$ space. As an application, we introduce a large class of novel financial asset pricing models that are based on polynomial jump-diffusions, and which are beyond the affine class. In these models, the excess log return processes are conditional L\'evy in the sense of \cite{cin_03}. This extends several well-known univariate diffusion volatility models, such as the Jacobi model \citep{ack_fil_pul_16}, the extended Stein--Stein model \citep{ste_ste_91,sch_zhu_99}, and the extended Hull--White model \citep{hul_whi_87,lio_mus_07}.

Due to their inherent tractability, polynomial jump-diffusions have played a prominent and growing role in a wide range of applications in finance. Examples include interest rates \citep{zho_03,del_shi:02,fil_lar_tro_17}, stochastic volatility \citep{gou_jas_06,ack_fil_pul_16}, exchange rates \citep{lar_sor_07}, life insurance liabilities \citep{bia_zha_16}, variance swaps \citep{fil_gou_man_16}, credit risk \citep{ack_fil_16}, dividend futures \citep{fil_wil_17}, commodities and electricity \citep{fil_lar_war_17}, stochastic portfolio theory \citep{cuc_17}, and economic equilibrium \citep{gua_won_18}. Properties of polynomial jump-diffusions can also be brought to bear on computational and statistical methods, such as generalized method of moments and martingale estimating functions \citep{for_sor_08}, variance reduction \citep{Cuchiero/etal:2012}, cubature \citep{fil_lar_pul_16}, and quantization \citep{cal_fio_pal_17}. This recent body of research primarily relies on polynomial jump-diffusions that are not necessarily affine. Focusing on the affine case, one finds an even richer history in the finance literature, where affine jump-diffusions have long been used to address a large number of problems in asset pricing, optimal investment, equilibrium analysis, etc.

In addition to their usefulness in applications, polynomial jump-diffusions are of theoretical interest due their rich mathematical structure. Work in this direction has primarily focused on the diffusion case, going back to \cite{Wong:1964}. \cite{Mazet:1997} and \cite{for_sor_08} characterize one-dimensional polynomial diffusions. \cite{Bakry/Orevkov/Zani:2014} describe those compact state spaces with nonempty interior that can support two-dimensional reversible polynomial diffusions. \cite{lar_pul_17} study polynomial diffusions on compact quadric sets, and relate their probabilistic properties to sum of squares representations of certain positive biquadratic forms. Existence and uniqueness for a large class of polynomial diffusions on semi-algebraic state spaces is developed by \cite{fil_lar_16}. Probability measure valued polynomial diffusions are studied by \citet{cuc_lar_sva_19}. The theoretical literature in the jump-diffusion case is less abundant. The first systematic accounts are \cite{Cuchiero:2011} and \cite{Cuchiero/etal:2012} in a Markovian framework. \cite{Gallardo/Yor:2006} study Dunkl processes, which are polynomial jump-diffusions; see also \cite{Dunkl:1992}. Affine jump-diffusions with compact state space are analyzed by \cite{kru_lar_18}. A large class of specifications of simplex-valued polynomial jump-diffusions is developed by \cite{cuc_lar_sva_18}. Our paper adds to this applied and theoretical literature by providing a unifying framework for polynomial and affine jump-diffusions, enabling new approaches to financial modeling.

We end this introduction with some conventions that will be used throughout this paper. We fix a stochastic basis $(\Omega,\Fcal,\Fcal_t,\P)$. Equalities between random variables are understood to hold almost surely. Following \citet[(I.1.1)]{Jacod/Shiryaev:2003} we use the notion of generalized conditional expectation, which is defined for any $\sigma$-field $\Fcal'\subset\Fcal$ and all random variables $X$ by
\[ \E[X\mid\Fcal'] =\begin{cases}
  \E[X^+\mid\Fcal'] -\E[X^-\mid\Fcal'] ,&\text{on $\{ \E[|X|\mid\Fcal']<\infty\}$,}\\ +\infty,&\text{elsewhere.}
\end{cases}\]
For a multi-index $\bm\alpha=(\alpha_1,\ldots,\alpha_d)\in\N^d_0$ we write $|\bm\alpha|=\alpha_1+\cdots+\alpha_d$ and $x^{\bm\alpha}=x_1^{\alpha_1}\cdots x_d^{\alpha_d}$ for $x\in\R^d$. A {\em polynomial}~$p$ on $\R^d$ is a function $p:\R^d\to\R$ of the form $p(x)=\sum_{\bm\alpha} c_{\bm\alpha} x^{\bm\alpha} $, where the sum runs over all $\bm\alpha\in\N^d_0$ and only finitely many of the coefficients $c_{\bm\alpha}$ are nonzero. Such a representation is unique. The {\em degree} of $p$ is the number $\deg p=\max\{ |\bm\alpha| : c_{\bm\alpha} \ne 0\}$. We let $\Pol(\R^d)$ denote the ring of all polynomials on $\R^d$, and $\Pol_n(\R^d)$ the subspace consisting of polynomials of degree at most~$n$. Let $E$ be a subset of~$\R^d$. A {\em polynomial on~$E$} is the restriction $p=q|_E$ to~$E$ of a polynomial $q\in\Pol(\R^d)$. Its degree is $\deg p=\min\{\deg q : p=q|_E,\, q\in\Pol(\R^d)\}$. We let $\Pol(E)$ denote the algebra of polynomials on~$E$,  and $\Pol_n(E)$ the subspace of polynomials on $E$ of degree at most~$n$. Both $\Pol_n(\R^d)$ and $\Pol_n(E)$ are finite-dimensional real vector spaces, but if there are nontrivial polynomials that vanish on~$E$ their dimensions will be different. If $E$ has a nonempty interior then $\Pol_n(\R^d)$ and $\Pol_n(E)$ can be identified. For simplicity of notation, we write $f\in\Pol_n(E)$ for any function $f=(f_1,\dots,f_k):\R^d\to\R^k$, with $k\in\N$, such that the restrictions ${f_i}|_E\in\Pol_n(E)$ for all $i=1,\dots,k$. The set of real symmetric $d\times d$ matrices is denoted $\S^d$, and the subset of positive semidefinite matrices is denoted~$\S^d_+$. For any map $\varphi:E\to F$, we denote the pullback operator by $\varphi^*$, which maps a function $f$ on $F$ to a function $\varphi^\ast f$ on $E$ by
\begin{equation} \label{eq:pullbackNEW}
\varphi^*f(x) = f(\varphi(x)).
\end{equation}

The remainder of the paper is as follows. In Section~\ref{secPJDs} we define polynomial jump-diffusions, give a characterization in terms of the coefficients of the extended generator, and establish the moment formula. In Section~\ref{secAJDs} we revisit affine jump-diffusions. In Section~\ref{secpolytrans} we study the invariance of the polynomial property under polynomial transformations. In Section~\ref{seccondLevy} we introduce polynomial conditional L\'evy processes. In Section~\ref{sectimechange} we show that the polynomial property is invariant under L\'evy time change. In Section~\ref{secpolyexp} we present a generic method for option pricing in polynomial jump-diffusion models. In Section~\ref{secPAPM} we introduce a large class of polynomial asset pricing models building on polynomial conditional L\'evy processes. In Section~\ref{secLVM} we focus on linear volatility models. Section~\ref{secconc} concludes. The appendix contains additional results and all proofs.

\section{Polynomial Jump-Diffusions} \label{secPJDs}

We consider a jump-diffusion operator on $\R^d$ of the form
\begin{equation} \label{eq:G}
\Gcal f(x) = \frac{1}{2}\tr(a(x)\nabla^2 f(x)) + b(x)^\top \nabla f(x) + \int_{\R^d} \left( f(x+\xi) - f(x) - \xi^\top \nabla f(x)\right)\nu(x,d\xi)
\end{equation}
for some measurable maps $a:\R^d\to\S^d_+$ and $b:\R^d\to\R^d$, and a transition kernel $\nu(x,d\xi)$ from $\R^d$ into $\R^d$ satisfying $\nu(x,\{0\})=0$ and $\int_{\R^d}\|\xi\|\wedge\|\xi\|^2\nu(x,d\xi)<\infty$ for all $x\in\R^d$.

We let $X_t$ be an $E$-valued jump-diffusion with extended generator $\Gcal$, for some state space $E\subseteq\R^d$. That is, $X_t$ is an $E$-valued special semimartingale and
\begin{equation} \label{eq:f-Gf loc mg_NEW}
f(X_t) - f(X_0) - \int_0^t \Gcal f(X_s)\,ds \quad\text{is a local martingale}
\end{equation}
for any bounded $C^2$ function $f(x)$ on $\R^d$.\footnote{The special semimartingale property of $X_t$ implies that
\[\int_0^t \left( \|a(X_s)\| + \|b(X_s)\| + \int_{\R^d} \|\xi\|^2\wedge\|\xi\|\,\nu(X_s,d\xi) \right) ds < \infty\]
and its semimartingale characteristics $(B,C,\nu)$ associated to the identity truncation function $h(\xi)=\xi$ are given by $B_t=\int_0^t b(X_s)ds$, $C_t=\int_0^t a(X_s)ds$, and $\nu(X_{t-},d\xi)dt$, see \citet[Proposition~II.2.29]{Jacod/Shiryaev:2003}. In particular, $X_t-B_t$ is a local martingale.}

We say that $\Gcal$ is {\em well-defined on $\Pol(E)$} if
\begin{gather}
\text{$  \int_{\R^d} \|\xi\|^{n}\,\nu(x,d\xi)<\infty$ for all $x\in E$ and all $n\ge 2$, and}\label{DefPol1} \\
\text{$\Gcal f(x)=0$ on $E$ for any $f\in \Pol(\R^d)$ with $f(x)=0$ on $E$.}\label{DefPol2}
\end{gather}
This property ensures that $\Gcal$ is well-defined as a linear operator from $\Pol(E)$ into the space of measurable functions on $E$.\footnote{Property~\eqref{DefPol2} always holds when $E$ has a nonempty interior. An example of a diffusion operator that is not well-defined on $\Pol(E)$ can be found in \citet[Section~2]{fil_lar_16}.}

\begin{defi} \label{D:PPg}
The operator $\Gcal$ is called {\em polynomial on $E$} if it is well-defined on $\Pol(E)$ and maps $\Pol_n(E)$ to itself for each $n\in\N$. In this case, we call $X_t$ a {\em polynomial jump-diffusion on $E$}.
\end{defi}

Thus, a polynomial jump-diffusion on $E$ is an $E$-valued jump-diffusion whose jump measure admits moments of all orders, with an extended generator that maps polynomials on $E$ to polynomials on $E$ of lower or equal degree. Polynomial jump-diffusions admit closed form conditional moments and have broad applications in finance, as we shall see below.

The polynomial property of $\Gcal$ on $E$ has a simple characterisation in terms of its coefficients $a(x)$, $b(x)$, and $\nu(x,d\xi)$.

\begin{lem} \label{LemPPchar}
Assume $\Gcal$ is well-defined on $\Pol(E)$. Then the following are equivalent:
\begin{enumerate}
\item\label{LemPPchar1} $\Gcal$ is polynomial on $E$;
\item\label{LemPPchar2} $a(x)$, $b(x)$, and $\nu(x,d\xi)$ in~\eqref{eq:G} satisfy
\begin{align*}
b  &\in \Pol_1(E),   \\
a  + \int_{\R^d} \xi\xi^\top \nu(\cdot ,d\xi) &\in \Pol_2(E),  \\
\int_{\R^d} \xi^{\bm \alpha}  \nu(\cdot, d\xi) &\in \Pol_{|{\bm \alpha}|}(E),
\end{align*}
for all $|{\bm\alpha}| \ge 3$.
\end{enumerate}
In this case, the polynomials on $E$ listed in property~{\ref{LemPPchar2}} are uniquely determined by the action of $\Gcal$ on $\Pol(E)$. Moreover, $a(x)$, $b(x)$, and $\int_{\R^d} \xi^{\bm \alpha}  \nu(x, d\xi)$ are locally bounded in $x$ on $E$ for all $|{\bm\alpha}| \ge 2$.
\end{lem}

\begin{proof}
The implication $\ref{LemPPchar2}\Rightarrow\ref{LemPPchar1}$ is immediate, and the implication $\ref{LemPPchar1}\Rightarrow\ref{LemPPchar2}$ follows by applying $\Gcal$ to all monomials. In particular, the polynomials on $E$ listed in property~{\ref{LemPPchar2}} are uniquely determined by the action of $\Gcal$ on $\Pol(E)$. It remains to show that $a(x)$ and $\int_{\R^d} \|\xi\|^2  \nu(x, d\xi)$ are locally bounded in $x$ on $E$. But this follows because $a(x)\in\S^d_+$ and $a_{ii} +\int_{\R^d} \xi_i^2  \nu(\cdot, d\xi)\in\Pol_2(E)$, and hence $a_{ii}(x)\ge 0$ and $\int_{\R^d} \xi_i^2  \nu(x, d\xi)\ge 0$, are locally bounded in $x$ on $E$.
\end{proof}

While the moments of $\nu(x,d\xi)$ of order three and beyond are determined by the action of $\Gcal$ on $\Pol(E)$, the measure $\nu(x,d\xi)$ itself need not be uniquely determined. The following example illustrates this.

\begin{ex}
Consider the compensated compound Poisson process $X_t$ with unit intensity and lognormal jump distribution, whose extended generator is given by $\Gcal f(x)=\int_\R (f(x+\xi)-f(x)-\xi f'(x))g(\xi)d\xi$, where $g(\xi)$ is the standard lognormal density. It is well-known that the lognormal distribution is indeterminate in the sense of the moment problem, so that there exists a density $h(\xi)$, different from $g(\xi)$, but with the same moments. The extended generator $\widetilde\Gcal f(x)=\int_\R (f(x+\xi)-f(x)-\xi f'(x))h(\xi)d\xi$ then coincides with $\Gcal$ on $\Pol(E)$.
\end{ex}

\begin{rem}
Non-Markovian polynomial jump-diffusions exist and can be constructed using the counterexamples of \citet[Section~3]{kal_kru_16}. This is why we take a semimartingale approach to polynomial jump-diffusions, rather than a Markovian approach as in \cite{Cuchiero/etal:2012}. The semimartingale approach has several advantages. First, we only have to assume existence, but do not require uniqueness, of the local martingale problem~\eqref{eq:f-Gf loc mg_NEW}. Indeed, uniqueness is tantamount to Markovianity of $X_t$. Furthermore, it is straightforward to develop a time-inhomogeneous version of our framework, where $a(t,x)$, $b(t,x)$, $\nu(t,x,d\xi)$ depend explicitly and not necessarily polynomially on $t$. Finally, our framework immediately accommodates polynomial jump-diffusions on finite time intervals. Indeed, we did not specify the time set for the stochastic basis and semimartingales above, so rather than $[0,\infty)$, one can simply understand the time set to be $[0,T]$ for some finite time horizon $T$. This flexibility is useful in asset pricing applications, and we make use of it in Section~\ref{S_measure_change}.
\end{rem}

To make the polynomial property of $\Gcal$ on $E$ operational, we now introduce a coordinate system on $\Pol_n(E)$, for a generic $n\in\N$, that we will use throughout the paper. We define $N=\dim\Pol_n(E)-1$ and note that $1+N\le {n+d \choose n}$, with equality if $E$ has nonempty interior. We fix polynomials $h_1(x),\dots,h_N(x)$ on $\R^d$ such that $\{1, h_1,\dots,h_N\}$ forms a basis of $\Pol_n(E)$, and define the vector valued function
\begin{equation}\label{CS1}
 H:\R^d\to\R^N,\quad H(x) = (h_1(x), \ldots, h_N(x))^\top.
\end{equation}
For each $p\in\Pol_n(E)$ we denote its coordinate vector by $\vec p\in\R^{1+N}$, so that
\begin{equation}\label{CS2}
p(x) = (1,H(x)^\top) \vec p \quad\text{on $E$.}
\end{equation}
The $(1+N)\times (1+N)$ matrix representation $G$ of $\Gcal$ restricted to $\Pol_n(E)$ is determined by $\Gcal (1,H^\top)(x)= (1,H(x)^\top) G$ on $E$, so that
\begin{equation}\label{CS3}
\Gcal p(x) = (1,H(x)^\top) G\vec p\quad\text{on $E$.}
\end{equation}

We now show that $\E[p(X_T) \mid \Fcal_t]$ is a polynomial function of~$X_t$.

\begin{thm}\label{T:moments}
Assume $\Gcal$ is polynomial on $E$. Then for any $p\in\Pol_n(E)$ the moment formula holds,
\[
\E[p(X_T) \mid \Fcal_t] = (1,H(X_t)^\top) \e^{(T-t)G}\,\vec p,\quad \text{for $t\le T$.}
\]
\end{thm}

Theorem~\ref{T:moments} implies that $X_T$ has finite $\Fcal_t$-conditional moments of all orders. Note that we do not assume that $X_T$ has any finite unconditional moments. The following example illustrates this.

\begin{ex}\label{exGARCH}
The GARCH diffusion $dX_t=\kappa(\theta-X_t)\,dt +\sqrt{2\kappa} X_t\,dW_t$ for some parameters $\kappa,\theta>0$ has a unique ergodic solution on $(0,\infty)$, which is a polynomial diffusion. The invariant distribution is an inverse Gamma distribution with shape parameter 2 and scale parameter $1/\theta$. Hence in the stationary case, when $X_0$ has the invariant distribution, we have $\E[X_t]=\theta$ and $\E[X_t^2]=+\infty$. See \citet[Case 4]{for_sor_08}.
\end{ex}

Here is a large class of polynomial jump-diffusions extending the GARCH diffusion. Let $W_t$ be a standard $m$-dimensional Brownian motion and $N(du,dt)$ a Poisson random measure with compensator $F(du)dt$ on $U\times \R_+$, for some mark space $U$, see \citet[Definition~II.1.20]{Jacod/Shiryaev:2003}. We consider the linear stochastic differential equation (SDE)
\begin{equation}\label{SDElin}
  dX_t = b(X_t)\,dt + \sigma(X_t)\,dW_t + \int_U \delta(X_{t-},u) (N(du,dt)-F(du)dt),
\end{equation}
with drift, volatility, and jump size functions
\begin{equation}\label{SDElincoeff}
  b(x)= \beta_0+\sum_{i=1}^d x_i\beta_i,\quad \sigma(x)= \Gamma_0 +\sum_{i=1}^d x_i\Gamma_i,\quad \delta(x,u)= \delta_0(u) + \sum_{i=1}^d x_i\delta_i(u),
\end{equation}
for parameters $\beta_i\in\R^{d}$, $\Gamma_i\in\R^{d\times m}$, and functions $\delta_i:U\to\R^d$ with $\int_U \|\delta_i(u)\|^n F(du)<\infty$ for all $n\ge 2$, for $i=0,\dots, d$. Due to the global Lipschitz continuity of the coefficients, the linear SDE~\eqref{SDElin} has a unique strong $\R^d$-valued solution $X_t$ for every $\Fcal_0$-measurable initial random variable $X_0$, see \citet[Theorem III.2.32]{Jacod/Shiryaev:2003}. It follows by inspection that $X_t$ is a polynomial jump-diffusion on $\R^d$ with linear drift $b\in\Pol_1(\R^d)$, diffusion function $a=\sigma\sigma^\top\in \Pol_2(\R^d)$, and jump measure $\nu(x,d\xi)$ given by $\int_{\R^d} f(\xi)\nu(x,d\xi) = \int_U f\left(\delta(x,u)\right)F(du)$, so that
\[ \int_{\R^d} \xi^{\bm\alpha}\nu(\cdot,d\xi) \in \Pol_{|\bm\alpha|}(\R^d)\quad \text{for all $|\bm\alpha|\ge 2$.}\]

We now derive the matrix representation \eqref{CS3} of $\Gcal$ for the linear SDE~\eqref{SDElin}--\eqref{SDElincoeff} in the univariate case, $d=1$, when $E\subseteq\R$ has nonempty interior. Straightforward calculations show that, for any $j\in \{0,\dots,n\}$,
\[ \Gcal x^j =\sum_{i=0}^n G_{ij} x^i \]
where
\[
G_{ij}=\begin{cases}
 {j \choose i}\int_U\delta_0(u)^{j-i}(1+\delta_1(u))^iF(du), & i\le j-3\\
 \frac{j(j-1)}{2}\Gamma_0^2 + \frac{j(j-1)}{2}\int_U\delta_0(u)^2(1+\delta_1(u))^{j-2}F(du), & i=j-2\\
 j\beta_0 + \frac{j(j-1)}{2}\Gamma_0\Gamma_1 + j\int_U \delta_0(u)\left((1+\delta_1(u))^{j-1}-1\right)F(du), & i=j-1\\
 j\beta_1 + \frac{j(j-1)}{2}\Gamma_1^2 + \int_U\left( (1+\delta_1(u))^j-1-j\delta_1(u)\right)F(du), & i=j\\
 0 ,& i>j.
 \end{cases}
 \]
Hence the $(1+n)\times (1+n)$ upper triangular matrix $G=(G_{ij})_{0\le i,j\le n}$ represents $\Gcal$ restricted to $\Pol_n(E)$ with respect to the basis $\{1,x,\dots,x^n\}$.

\begin{ex}
 A special case of the linear SDE~\eqref{SDElin}--\eqref{SDElincoeff} is the L\'evy driven SDE
\[ dX_t = b(X_t)\,dt + \sigma(X_{t-})\,dL_t \]
where $b(x)$ and $\sigma(x)$ are as in \eqref{SDElincoeff}, and $L_t$ is an $m$-dimensional L\'evy process with $\E[\|L_t\|^n]<\infty$ for all $n\ge 2$, see \citet[Theorem 25.3]{sat_99}.
\end{ex}

\section{Affine Jump-Diffusions} \label{secAJDs}

Affine jump-diffusions have been widely studied and applied in finance, see e.g.\ \cite{duf_fil_sch_03}. We provide a novel  approach to affine jump-diffusions and show that they constitute examples of polynomial jump-diffusions.

Let $\Gcal$ be a jump-diffusion operator on $\R^d$ of the form \eqref{eq:G}, and let $X_t$ be an $E$-valued jump-diffusion with extended generator $\Gcal$, for some state space $E\subseteq\R^d$.

\begin{defi} \label{D:affine}
The operator $\Gcal$ is called {\em affine on $E$} if there exist complex-valued functions $F(u)$ and $R(u)=(R_1(u),\dots,R_d(u))^\top$ on ${\rm i}\R^d$ such that
\begin{equation} \label{eq:D:affine}
\Gcal \e^{ u^\top x} = \left( F(u) + R(u)^\top x\right)\e^{ u^\top x}
\end{equation}
holds for all $x\in E$ and $u\in {\rm i}\R^d$. In this case, we call $X_t$ an {\em affine jump-diffusion on $E$}.
\end{defi}

Note that this is a relaxed definition compared to the definition of an affine process in \cite{duf_fil_sch_03}, because it is directly given in terms of the point-wise action of $\Gcal$ on exponential-affine functions. Indeed, the affine jump-diffusion in Example~\ref{exaffinenotDFS} below is not an affine process in the sense of \cite{duf_fil_sch_03}.

In analogy to Lemma~\ref{LemPPchar}, the affine property of $\Gcal$ on $E$ has a simple characterization in terms of its coefficients $a(x)$, $b(x)$, and $\nu(x,d\xi)$.

\begin{lem}\label{lemaffine}
The operator $\Gcal$ is affine on $E$ if and only if $a(x)$, $b(x)$, and $\nu(x,d\xi)$ are affine of the form
\begin{equation} \label{eq:affine abnu}
a(x) = a_0 + \sum_{i=1}^d x_i a_i, \quad b(x) = b_0 + \sum_{i=1}^d x_i b_i, \quad \nu(x,d\xi) = \nu_0(d\xi) + \sum_{i=1}^d x_i \nu_i(d\xi),\quad\text{on $E$}
\end{equation}
for some matrices $a_i\in\S^d$, vectors $b_i\in\R^d$, and signed measures $\nu_i(d\xi)$ on $\R^d$ such that $\nu_i(\{0\})=0$ and $\int_{\R^d} \|\xi\|\wedge\|\xi\|^2\, |\nu_i|(d\xi)<\infty$, $i=0,\ldots,d$. In this case, the functions $F(u)$ and $R(u)$ in~\eqref{eq:D:affine} can be chosen to be of the form
\begin{equation}\label{eq:affine Ri}
\begin{aligned}
F(u)& = \frac{1}{2}u^\top a_0 u + b_0^\top u + \int_{\R^d} \left( e^{u^\top \xi} - 1 - u^\top\xi \right) \nu_0(d\xi),\\
R_i(u) &= \frac{1}{2}u^\top a_i u + b_i^\top u + \int_{\R^d} \left( e^{u^\top \xi} - 1 - u^\top\xi \right) \nu_i(d\xi).
\end{aligned}
\end{equation}
\end{lem}

It follows from \eqref{eq:affine abnu} that an affine jump-diffusion $X_t$ cannot be realized as solution of a linear SDE~\eqref{SDElin}--\eqref{SDElincoeff} in general unless diffusion and jump coefficients $a(x)=a_0$ and $\nu(x,d\xi)=\nu_0(d\xi)$ do not depend on $x$.

From \eqref{eq:affine abnu} and Lemma~\ref{LemPPchar} we immediately obtain that affine jump-diffusions are polynomial on $E$, subject to being well-defined on $\Pol(E)$.

\begin{cor}
If $X_t$ is an affine jump-diffusion on $E$ such that $\Gcal$ is well-defined on $\Pol(E)$, then $X_t$ is a polynomial jump-diffusion on $E$.
\end{cor}

Affine jump-diffusions on $E$ not only satisfy the moment formula in Theorem~\ref{T:moments}, subject to being well-defined on $\Pol(E)$. Their characteristic functions are also analytically tractable.

\begin{thm}\label{thmATF}
Assume $X_t$ is an affine jump-diffusion on $E$. For $u\in{\rm i}\R^d$ and $T>0$, let $\phi(\tau)$ and $\psi(\tau)=(\psi_1(\tau),\dots,\psi_d(\tau))^\top$ be functions that solve the generalized Riccati equations
\begin{equation} \label{eq:Riccati}
\begin{aligned}
\phi'(\tau) 	&= F(\psi(\tau)),	\quad \phi(0) =0\\
\psi'(\tau) 	&= R(\psi(\tau)),	\quad \psi(0) =u
\end{aligned}
\end{equation}
for $0\le\tau\le T$, where $F(u)$ and $R(u)$ are the functions in~\eqref{eq:affine Ri}. If
\begin{equation} \label{eq:Re bound}
{\rm Re\,} \phi(T-t) +  {\rm Re\,} \psi(T-t)^\top X_{t} \le 0, \quad \text{for $t\le T$,}
\end{equation}
then the affine transform formula holds,
\[
\E[\e^{u^\top X_T} \mid \Fcal_t] = \e^{\phi(T-t) + \psi(T-t)^\top X_{t}}, \quad \text{for $t\le T$.}
\]
\end{thm}

\begin{rem}
Inequality~\eqref{eq:Re bound} is a necessary condition for the affine transform formula to hold. Indeed, the affine transform formula and Jensen's inequality yield, for $u\in{\rm i}\R^d$,
\[
\exp\left( {\rm Re\,} \phi(T-t) +  {\rm Re\,} \psi(T-t)^\top X_{t} \right) = \left| \E[\exp(u^\top X_T) \mid \Fcal_t] \right| \le \E[ |\exp(u^\top X_T)| \mid \Fcal_t] = 1.
\]
\end{rem}

There exist affine jump-diffusions for which the generalized Riccati equations~\eqref{eq:Riccati} do not admit global solutions for all $u\in{\rm i}\R^d$. The following example illustrates this.

\begin{ex}\label{exaffinenotDFS}
Consider the two-point state space $E=\{0,1\}\subseteq\R$, and the process $X_t$ that jumps from~$1$ to~$0$ with intensity $\lambda$ and is absorbed once it reaches~$0$. This is a jump-diffusion with extended generator
\[
\Gcal f(x) = \lambda x (f(x-1)-f(x)),
\]
which is of the form~\eqref{eq:G} with $a(x)=0$, $b(x)=\lambda x$, $\nu(x,d\xi)=\lambda x\delta_{-1}(d\xi)$. Thus $X_t$ is an affine jump-diffusion, and $F(u)=0$ and $R(u)=\lambda(e^{-u}-1)$. The associated generalized Riccati equation~\eqref{eq:Riccati} is
\begin{equation} \label{ex_Ric_eq1}
\phi'(\tau)=0,\quad \psi'(\tau)=\lambda(e^{-\psi(\tau)}-1).
\end{equation}
We claim that this equation does not have a global solution for the initial condition $u={\rm i}\pi$. We argue by contradiction and assume that $\psi(\tau)$ is a global solution of~\eqref{ex_Ric_eq1}. Then $\Psi(\tau)=e^{\psi(\tau)}$ satisfies the linear equation
\[
\Psi'(\tau)=-\lambda\Psi(\tau)+\lambda, \qquad \Psi(0)=-1,
\]
whose unique solution is $\Psi(\tau)=1-2e^{-\lambda\tau}$, which becomes zero for $\tau=\lambda^{-1}\log 2$, which is absurd.

The deeper reason behind this fact is that the characteristic function of $X_T$ given $X_0=1$, $\E[\e^{u X_T}]=1-\e^{-\lambda T}+\e^{u -\lambda T}$, for $u={\rm i}\pi$ and $T=\lambda^{-1}\log 2$ becomes zero and hence cannot be written as exponential as in the affine transform formula.
\end{ex}

\section{Polynomial Transformations}\label{secpolytrans}

The class of polynomial jump-diffusions is shown to be invariant under polynomial transformations, after an extension of the dimension. This allows us to build a large class of polynomial jump-diffusions from basic building blocks, including Brownian motion, L\'evy processes, or more general affine processes. This turns out to be a useful and flexible method for introducing non-linearities and jumps in all kinds of financial models.

Let $X_t$ be a polynomial jump-diffusion on $E\subseteq\R^d$ with extended generator $\Gcal$. Fix $n\in\N$ and a basis $\{1,h_1,\dots,h_N\}$ of $\Pol_n(E)$ as in~\eqref{CS1}--\eqref{CS3}. Notice that that $H:E\to H(E)\subseteq\R^N$ is injective. Indeed, the restriction to $E$ of any linear monomial $x_i$ lies in $\Pol_n(E)$, and is therefore a linear combination of $1,h_1(x),\dots,h_N(x)$. Thus there exist linear polynomials $\ell_i\in\Pol_1(\R^N)$ such that $\ell_i(H(x)) = x_i$ for all $x\in E$ and all $i$. Define the vector valued function
\begin{equation}\label{Ldef}
 L:\R^N\to\R^d,\quad L(\overline x) = (\ell_1(\overline x), \ldots, \ell_d(\overline x))^\top.
\end{equation}
Then $L(H(x))=x$ for all $x\in E$, and $H(L(\overline x))=\overline x$ for all $\overline x\in H(E)$. We define the pullbacks  $H^*$ and $L^*$ as in~\eqref{eq:pullbackNEW}.

\begin{lem}\label{lemPBHL}
For every $m\in\N$, the pullback $H^\ast:\Pol_{m}(H(E))\to\Pol_{mn}(E)$ is a linear isomorphism with inverse $L^\ast$.
\end{lem}

Here is the main result of this section.\footnote{This result was derived in collaboration with Sergio Pulido, and is applied in \cite{fil_lar_pul_16}.}

\begin{thm} \label{T:exp}
The process $\overline X_t=H(X_t)$ is a polynomial jump-diffusion on $H(E)$ with extended generator $\overline\Gcal = L^\ast\Gcal H^\ast$ and, for every $m\in\N$, the following diagram commutes:
\begin{equation}\label{diagXX}
\begin{tikzcd}
\Pol_m(H(E)) \arrow{r}{\overline\Gcal} \arrow[swap]{d}{H^*}& \Pol_m(H(E)) \\
\Pol_{mn}(E) \arrow{r}{\Gcal}  & \Pol_{mn}(E) \arrow[swap]{u}{L^*}
\end{tikzcd}
\end{equation}
\end{thm}

As an immediate application of Theorem~\ref{T:exp} we can easily infer the action of $\overline \Gcal$ on $\Pol_m(H(E))$. Let $1+\overline N=\dim\Pol_{m}(H(E)) =\dim\Pol_{mn}(E)$ and extend the basis of $\Pol_n(E)$ to a basis
\begin{equation}\label{basisHE}
  \{h_0=1, h_1,\dots,h_N ,h_{N+1},\dots, h_{\overline N}\}
\end{equation}
of $\Pol_{mn}(E)$, for some polynomials $h_{N+1}(x),\dots, h_{\overline N}(x)$ on $\R^d$. In view of the commuting diagram~\eqref{diagXX} this induces a basis $\overline h_i = L^\ast h_i$ on $\Pol_m(H(E))$, for $i=0,\dots,\overline N$. Let $\overline G$ be the $(1+\overline N)\times (1+\overline N)$ matrix representing $\Gcal$ on $\Pol_{mn}(E)$ according to \eqref{CS1}--\eqref{CS3}, which can be determined using symbolic calculus applied to $\Gcal h_i(x)$. Then $\overline G$ is the matrix representing $\overline\Gcal$ on $\Pol_m(H(E))$, and Theorem~\ref{T:moments} can readily be applied to compute all $\Fcal_t$-conditional moments of $\overline X_T$ up to order $m$.

The following example shows that the affine property is not invariant under polynomial transformations.

\begin{ex}\label{exSQRnotaffine}
Consider the square-root process $dX_t= (b+\beta X_t)\,dt+  \sigma \sqrt{X_t}\,dW_t$, which is an affine diffusion. The augmented process $\overline X_t=(X_t,X_t^2)$ satisfies
\begin{align*}
d\overline X_{1t} &=  (b+\beta \overline X_{1t})\,dt+  \sigma\sqrt{\overline X_{1t}} \,dW_t \\
d\overline X_{2t} &= ((2b+\sigma^2) \overline X_{1t}+2\beta \overline X_{2t}) \,dt + 2\sigma  \sqrt{\overline X_{1t}\overline X_{2t}} \,dW_t.
\end{align*}
While the drift function of $\overline X_t$ is affine of the form \eqref{eq:affine abnu}, the diffusion function is not. In view of Lemma~\ref{lemaffine} this shows that $\overline X_t$ is not affine, while it is polynomial, on $H(\R)$ for $H(x)=(x,x^2)^\top$, in line with Theorem~\ref{T:exp}.
\end{ex}

\section{Polynomial Conditional L\'evy Processes} \label{seccondLevy}

In financial applications we often encounter the following situation. Let $X_t$ be a polynomial jump-diffusion $X_t$ on $E\subseteq\R^d$, and let $Y_t$ be an $\R^e$-valued semimartingale, for some $e\in\N$, whose characteristics are functions of $X_t$. The process $Y_t$ could model the excess log returns of assets whose stochastic volatilities and jump characteristics are given in terms of the latent factor process $X_t$. Drawing on Section~\ref{secpolytrans} we develop a polynomial framework that accommodates a large class of such models. The following example illustrates the kind of situation we have in mind. We elaborate on this example further, including jumps, in Section~\ref{secLVM}.

\begin{ex}\label{ex_5_1}
For $\gamma>0$, $\kappa\theta>0$, $X_0>0$, and $Y_0=0$, we consider the following model specified under the risk-neutral measure:
\begin{align*}
dX_t &= \kappa(\theta-X_t)\,dt + \gamma X_t\,dW_{1t}, \\
dY_t &= -\frac12 X_t^2\,dt + X_t\,dW_{2t},
\end{align*}
where $Y_t$ models the excess log return of an asset and $X_t$ its volatility. In view of Lemma~\ref{LemPPchar} and Example~\ref{exGARCH}, $X_t$ is a polynomial diffusion on $E=(0,\infty)$. Moreover, the drift and diffusion functions of $Y_t$ are both quadratic in $X_t$. In particular, Lemma~\ref{LemPPchar} shows that the joint process $Z_t=(X_t,Y_t)$ is not polynomial on $E\times\R$. However, the augmented process $\overline Z_t =(H(X_t),Y_t)$ with $H(x)=(x,x^2)$ is a polynomial diffusion on $H(E)\times\R$. Thus the moment formula in Theorem~\ref{T:moments} can still be used to compute conditional moments of $Y_T$.
\end{ex}

Returning to the general discussion, we assume that the joint semimartingale $Z_t=(X_t,Y_t)$ is an $E\times\R^e$-valued jump-diffusion with extended generator of the form
\begin{equation} \label{eq:GZ_19}
\Gcal f(z) = \frac{1}{2}\tr(a(x)\nabla^2 f(z)) + b(x)^\top \nabla f(z) + \int_{\R^{d+e}} \left( f(z+\zeta) - f(z) - \zeta^\top \nabla f(z)\right)\nu(x,d\zeta),
\end{equation}
where we write $z=(x,y)$, for some measurable maps $a:\R^d\to\S^{d+e}_+$ and $b:\R^d\to\R^{d+e}$, and a transition kernel $\nu(x,d\zeta)$ from $\R^d$ into $\R^{d+e}$ satisfying $\nu(x,\{0\})=0$ and $\int_{\R^{d+e}}\|\zeta\|\wedge\|\zeta\|^2\nu(x,d\zeta)<\infty$ for all $x\in\R^d$.

According to the decomposition $Z_t=(X_t,Y_t)$, and accordingly $\zeta=(\xi,\eta)$, we write
\begin{equation}\label{aXaY_19}
 a(x)=\begin{pmatrix}
 a^{X}(x) & a^{XY}(x) \\ a^{YX}(x) & a^{Y}(x))
\end{pmatrix},\quad b(x)=\begin{pmatrix}b^X(x) \\ b^Y(x)
\end{pmatrix},\quad \nu(x,d\zeta)=\nu(x,d\xi\times d\eta),
\end{equation}
and denote by $\nu^X(x,d\xi)$ and $\nu^Y(x,d\eta)$ the marginal measures of $\nu(x,d\xi\times d\eta)$ given by
\begin{equation}\label{nuXnuY_19}
 \nu^X(x,A) =   \nu(x,A\times\R^e),\quad  \nu^Y(x,A) =   \nu(x,\R^d\times A).
\end{equation}
Then $a^X(x)$, $b^X(x)$, and $\nu^X(x,d\xi)$ are the coefficients of the extended generator $\Gcal^X$ of $X_t$, which is polynomial on $E$ by assumption. Note that $Y_t$ is a conditional L\'evy process in the sense of \cite{cin_03}. That is, conditional on the process $X_t$, the semimartingale $Y_t$ has independent increments.

Fix $n\in\N$ and a basis $\{1,h_1,\dots,h_N\}$ of $\Pol_n(E)$ as in~\eqref{CS1}--\eqref{CS3} for $\Gcal^X$ in lieu of $\Gcal$. Extending \eqref{Ldef}, we define the maps $\varphi:\R^{d+e}\to\R^{N+e}$ and $\psi:\R^{N+e}\to\R^{d+e}$ by
\[   \varphi(x,y) = (H(x),y),\quad \psi(\overline x,y) = (L(\overline x),y),\]
so that $\psi\circ\varphi=\id$ on $E\times\R^e$. The pullbacks  $\varphi^*$ and $\psi^*$ are defined in~\eqref{eq:pullbackNEW}.  Here is our first main result of this section, which extends Example~\ref{ex_5_1}.

\begin{thm} \label{T:augNEW1_19}
Assume that
\begin{equation}\label{DefPolaug1_19}
 \text{$  \int_{\R^{e}} \|\eta\|^{k}\,\nu^Y(x,d\eta)<\infty$ for all $x\in E$ and all $k\ge 2$.}
 \end{equation}
Then the augmented process $\overline Z_t =(H(X_t),Y_t)$ is a jump-diffusion on $H(E)\times\R^e$ with extended generator $\overline\Gcal = \psi^\ast\Gcal \varphi^\ast$, and the operators $\Gcal$ and $\overline\Gcal$ are well-defined on $\Pol(E\times\R^e)$ and $\Pol(H(E)\times\R^e)$, respectively. Furthermore, the properties
\begin{align}
b^Y  &\in \Pol_{n}(E),   \label{eq:T:aug:bYNew_19} \\
a^{Y}   + \int_{\R^{e}} \eta\eta^\top \nu^Y(\cdot,d\eta) &\in \Pol_{2n}(E),  \label{eq:T:aug:aYYNew_19} \\
a^{XY}  + \int_{\R^{d+e}} \xi\eta^\top \nu(\cdot,d\xi\times d\eta) &\in \Pol_{1+n}(E), \label{eq:T:aug:aXYNew_19} \\
\int_{\R^{d+e}} \xi^{\bm \alpha} \eta^{\bm \beta} \nu(\cdot, d\xi\times d\eta) &\in \Pol_{|{\bm\alpha}|+n|\bm\beta|}(E), \quad\text{for all $|{\bm\alpha}|+|{\bm\beta}|\ge 3$,} \label{eq:T:aug:nuew_19}
\end{align}
together imply
\begin{equation}\label{T_augNEW1_19_PJD}
\text{$\overline Z_t$ is polynomial on $H(E)\times\R^e$.}
\end{equation}
Conversely, \eqref{T_augNEW1_19_PJD} implies \eqref{eq:T:aug:bYNew_19}, \eqref{eq:T:aug:aYYNew_19}, and \eqref{eq:T:aug:nuew_19} for $\bm\alpha=0$.\footnote{We conjecture that properties~\eqref{eq:T:aug:aXYNew_19} and \eqref{eq:T:aug:nuew_19} are not necessary for \eqref{T_augNEW1_19_PJD} to hold in general. However, we have not found a counterexample.}
\end{thm}

\begin{rem}
Since $\Gcal^X$ is polynomial on $E$, and thus well-defined on $\Pol(E)$, condition \eqref{DefPolaug1_19} is equivalent to
\[ \text{$  \int_{\R^{d+e}} \|\zeta\|^{k}\,\nu(x,d\zeta)<\infty$ for all $x\in E$ and all $k\ge 2$.}\]
\end{rem}

As an application of Theorem~\ref{T:augNEW1_19} we show how to construct large classes of polynomial jump-diffusions by specifying $Y_t$ in terms of the polynomial jump-diffusion $X_t$ on $E$.

\begin{cor}\label{corPQdX_19}
Let $e=e'+e''$ for some $e',e''\ge 0$, and consider the maps $P:E\to\R^{e'}$ and $Q:E\to\R^{e''\times d}$ with polynomial components, $P \in\Pol_n(E)$ and $Q \in\Pol_{n-1}(E)$. Then for
\[ dY_t = \begin{pmatrix}
  P(X_t)\,dt \\ Q(X_{t-})\,dX_t
\end{pmatrix}\]
the conditions \eqref{DefPolaug1_19}--\eqref{eq:T:aug:nuew_19} in Theorem~\ref{T:augNEW1_19} are satisfied, so that $\overline Z_t=(H(X_t),Y_t)$ is a polynomial jump-diffusion on $H(E)\times\R^e$.

For $n\ge 2$, this covers the quadratic co-variations, $d[X_i,X_j]_t=d(X_{i,t}X_{j,t}) - X_{i,t-}dX_{j,t} - X_{j,t-}dX_{i,t}$, and their predictable compensators, $\Gamma^X(x_i,x_j)(X_t)\,dt$, where $\Gamma^X$ denotes the carr\'e-du-champ operator related to $\Gcal^X$ (see Section~\ref{S:cdc}).
\end{cor}

To compute conditional moments of $Y_T$ using the moment formula in Theorem~\ref{T:moments}, we must understand the structure of $\Pol_m(H(E)\times\R^e)$ and how $\overline\Gcal$ acts on it. To this end, we introduce the subspace $V_m\subseteq\Pol_{nm}(E\times\R^e)$ defined by
\begin{equation}\label{eq:VmNew_19}
V_m = \vspan\{ p(x)y^{\bm\beta} \colon p\in\Pol(E),\  \deg p \le n(m- |{\bm\beta}|),\ |{\bm\beta}| \le m\}.
\end{equation}
Extending Lemma~\ref{lemPBHL} we have the following result.
\begin{lem}\label{lemPBHL2_19}
For every $m\in\N$, the pullback $\varphi^\ast:\Pol_{m}(H(E)\times\R^e)\to V_m$ is a linear isomorphism with inverse $\psi^\ast$.
\end{lem}

Here is our second main result of this section.

\begin{thm} \label{T:augNEW2_19}
Assume~\eqref{DefPolaug1_19}. Then either of the following statements is equivalent to \eqref{T_augNEW1_19_PJD}:
\begin{enumerate}

\item\label{T:augNEW12_19} The operator $\Gcal$ maps the space $V_m$ to itself for each $m\in\N$;

\item\label{T:augNEW13_19} $\Gcal(y^{\bm\beta})$ and $\Gamma(x^{\bm\alpha},y^{\bm\beta})$ lie in $V_m$ whenever $|\bm\alpha|\le n(m-|\bm\beta|)$ and $|\bm\beta|\le m$, where $\Gamma$ denotes the carr\'e-du-champ operator related to $\Gcal$ (see Section~\ref{S:cdc}).
\end{enumerate}
In either case, for every $m\in\N$, the following diagram commutes:
\begin{equation}\label{diagXX2_19}
\begin{tikzcd}
\Pol_m(H(E)\times\R^e) \arrow{r}{\overline\Gcal} \arrow[swap]{d}{\varphi^*}& \Pol_m(H(E)\times\R^e) \\
V_m \arrow{r}{\Gcal}  & V_m \arrow[swap]{u}{\psi^*}
\end{tikzcd}
\end{equation}
\end{thm}

\begin{rem}
Note that for $n=1$ and $H(x)=x$ we have $\overline Z_t=Z_t$ and $V_m=\Pol_{m}(E\times\R^e)$, in which case Theorem~\ref{T:augNEW2_19} simply recovers the definition of $Z_t$ being polynomial on $E\times\R^e$.
\end{rem}

As an application of Theorem~\ref{T:augNEW2_19} we can infer the action of $\overline \Gcal$ on $\Pol_m(H(E)\times\R^e)$. This allows us to compute conditional moments of $Y_T$ using the moment formula in Theorem~\ref{T:moments}. Assume \eqref{DefPolaug1_19} and \eqref{T_augNEW1_19_PJD}. Let $1+ \overline N= \dim\Pol_{mn}(E)$ and extend the basis of $\Pol_n(E)$ to a basis of $\Pol_{mn}(E)$ as in \eqref{basisHE}. This induces a basis of $V_m$ of the form
\[ v_i(x,y) =  h_j(x)y^{\bm \beta},\quad   \deg h_j \le n(m-|\bm\beta|),\quad |\bm\beta|\le m , \quad i=0,\dots, M \]
where $1+ M=\dim V_m$. In view of the commuting diagram~\eqref{diagXX2_19} this induces a basis $\overline v_i = \psi^\ast v_i$, $i=0,\dots,  M$, of $\Pol_m(H(E)\times\R^e)$. Let $G$ be the $(1+ M)\times (1+ M)$ matrix representing $\Gcal$ on $V_m$ according to \eqref{CS1}--\eqref{CS3} with $v_i$ in lieu of $h_i$, which can be determined using symbolic calculus applied to $\Gcal v_i(x,y)$ for $v_i(x,y)=  h_j(x)y^{\bm \beta}$. Then $ G$ is the matrix representing $\overline\Gcal$ on $\Pol_m(H(E)\times\R^e)$, and Theorem~\ref{T:moments} can readily be applied to compute all conditional moments of $\overline Z_T$ up to order $m$.

The following corollary is useful for applications because it helps to reduce the dimension for moment computations. For example, if we only need the conditional moments of $Y_{1T}$, this does not involve the remaining components $Y_{it}$ for $i\neq 1$.
\begin{cor}\label{cordimred_19}
Assume \eqref{DefPolaug1_19} and \eqref{T_augNEW1_19_PJD}. Let $P:\R^e\to \R^{e'}$ be a linear map, for some $e'\in\N$. Then $Z_t'=(X_t, PY_t)$ satisfies \eqref{DefPolaug1_19} and \eqref{T_augNEW1_19_PJD} in lieu of $Z_t=(X_t,Y_t)$ with dimension $e$ replaced by $e'$.
\end{cor}

\section{L\'evy Time Change}\label{sectimechange}

The class of polynomial jump-diffusions is shown to be invariant under time change by a L\'evy subordinator. Let $X_t$ be a polynomial jump-diffusion on $E\subseteq\R^d$ with extended generator $\Gcal$. Let $Z_t$ be an independent nondecreasing L\'evy process (subordinator) with L\'evy measure $\nu^Z(d\zeta)$ and drift $b^Z\ge0$ so that its generator is
\[
\Gcal^Z f(z) = b^Z f'(z) + \int_0^\infty \left( f(z+\zeta) - f(z) \right) \nu^Z(d\zeta).
\]
The law of $Z_t$ is denoted by $\mu^t(dz)$. A heuristic argument suggests that the time-changed process $\widetilde X_t = X_{Z_t}$ is again a polynomial jump-diffusion on $E$. Indeed, the moment formula Theorem~\ref{T:moments}, the independence of $X_t$ and $Z_t$, and the L\'evy property of $Z_t$ give, for any polynomial $f(x)=(1,H(x)^\top) \vec f$ in $\Pol_n(E)$,
\begin{align*}
 \E[ f(\widetilde X_T) \mid \widetilde X_t ] &=  \E[ \E[ f(X_{Z_T}) \mid   X_{Z_t},Z_t,Z_T ] \mid X_{Z_t} ]=\E[(1, H(X_{Z_t})^\top)   \e^{(Z_T-Z_t)G}  \vec f \mid X_{Z_t} ]  \\
 &=(1, H(X_{Z_t})^\top) \int_0^\infty \e^{zG} \mu^{T-t}(dz)\vec f =(1, H(\widetilde X_t)^\top) \e^{(T-t)\widetilde G}\vec f,
\end{align*}
where the matrix $\widetilde G$ is given in \eqref{eqGGtilde} below, subject to $\mu^t(dz)$-integrability conditions. Hence $\widetilde X_t$ satisfies the moment formula. However, it turns out to be surprisingly difficult, if not impossible, to prove without any further assumptions that $\widetilde X_t$ is a jump-diffusion.\footnote{It is straightforward to prove that $\widetilde X_t$ is a semimartingale, but it is not clear that its jump characteristic is a function of the current state only, as required for jump-diffusions. What is more, the drift and jump characteristics in \eqref{sub1:btilde} and \eqref{sub1:nutilde} are given in terms of the Markov transition kernel of $X_t$.} Assuming Markovianity, we can prove the following result.

\begin{thm} \label{T:sub}
Assume that $X_t$ is a Feller process with transition kernel $p_t(x,dy)$, the domain of its generator contains $C^\infty_c(E)$, and the generator coincides with $\Gcal$ on $C^\infty_c(E)$. Assume also that the L\'evy measure $\nu^Z(d\zeta)$ admits exponential moments,
\begin{equation} \label{eq:nuZ exp moments}
\int_1^\infty e^{\zeta \lambda}  \nu^Z(d\zeta)<\infty,
\end{equation}
for any $\lambda$ from the set of real parts of eigenvalues of $\Gcal$ restricted to $\Pol(E)$. Then the time-changed process $\widetilde X_t = X_{Z_t}$ is a polynomial jump-diffusion on $E$ and a Feller process with transition kernel
\begin{equation}\label{deftildept}
\widetilde p_t(x,dy) = \int_0^\infty  p_z(x,dy)\mu^t(dz)
\end{equation}
with respect to the usual augmentation $\widetilde\Fcal_t$ of its natural filtration. Its extended generator
\begin{equation} \label{eq:sub1_Gtilde}
\begin{aligned}
  \widetilde\Gcal f(x) &= \frac{1}{2}\tr(\widetilde a(x)\nabla^2 f(x)) + \widetilde b(x)^\top \nabla f(x) \\
&\quad + \int_{\R^d} \left( f(x+\xi) - f(x) - \xi^\top \nabla f(x)\right)\widetilde\nu(x,d\xi)
\end{aligned}
\end{equation}
is given by
\begin{align}
\widetilde a(x) &= b^Z a(x), \label{sub1:atilde} \\
\widetilde b(x) &= b^Z b(x) + \int_0^\infty \int_E (y-x) p_\zeta(x,dy) \nu^Z(d\zeta), \label{sub1:btilde} \\
\widetilde\nu(x,d\xi) &= b^Z \nu(x,d\xi)+ \int_0^\infty  1_{\{\xi\neq 0\}} p_\zeta(x,x+d\xi)  \nu^Z(d\zeta). \label{sub1:nutilde}
\end{align}
The matrix representations $G$ and $\widetilde G$ of $\Gcal$ and $\widetilde\Gcal$ on $\Pol_n(E)$ are related by
\begin{equation}\label{eqGGtilde}
 \widetilde G = b^Z G + \int_0^\infty (\e^{\zeta G} - \id) \nu^Z(d\zeta)\quad\text{and}\quad \e^{t\widetilde G} = \int_0^\infty \e^{z G}\mu^{t}(dz).
\end{equation}

\end{thm}

\begin{rem}\label{remJM}
It is shown in the proof of Theorem~\ref{T:sub} that $p_t(x,A)$ and $p_t(x,x+A)$ are jointly measurable in $(t,x)\in [0,\infty)\times E$ for all measurable $A\subseteq\R^d$, with the obvious extension $p_t(x,A)=p_t(x,A\cap E)$. Hence \eqref{deftildept} and \eqref{sub1:nutilde} specify well defined transition kernels from $\R^d$ into $\R^d$, defined to be zero for $x\notin E$.
\end{rem}

\begin{rem}\label{rembtilde}
Theorem~\ref{T:moments} yields that
\[ \int_E (y-x) ^\top p_\zeta(x,dy) = (1,H(x)^\top)  ( \e^{\zeta G }-\id) M   \]
where $G$ is the matrix representation of $\Gcal$ on $\Pol_1(E)$ and $M$ is the matrix whose $i$th column is the corresponding vector representation of $x_i$ in $\Pol_1(E)$. In view of \eqref{eq:nuZ exp moments} it thus follows that \eqref{sub1:btilde} specifies a well defined first order polynomial drift function.
\end{rem}

\begin{rem}\label{remLK}
\citet[Theorem 25.3]{sat_99} states that condition~\eqref{eq:nuZ exp moments} is equivalent to $\E[\e^{\lambda Z_t}]=\int_0^\infty \e^{z\lambda} \mu^t(dz)<\infty$ for all $t\ge 0$. Hence the integrals in \eqref{eqGGtilde} are well defined. If $E$ is compact then all eigenvalues of $\Gcal$ restricted to $\Pol(E)$ have nonpositive real part, such that \eqref{eq:nuZ exp moments} trivially holds.
\end{rem}

The following example shows that the affine property is not invariant with respect to L\'evy time change.

\begin{ex}\label{exLTCOU}
Consider the Ornstein--Uhlenbeck process $dX_t = -\kappa X_t\,dt + \sigma\,dW_t$, which is an affine Feller process with normal transition kernel $p_t(x,dy)$ with mean $\e^{-\kappa t}x$ and variance $\frac{\sigma^2}{2\kappa}\left(1-\e^{-2\kappa t}\right)$. Now consider an independent Poisson subordinator $Z_t$ with $\beta^Z=0$ and $\nu^Z(d\zeta)=\delta_{\{1\}}(d\zeta)$. According to Theorem~\ref{T:sub}, the L\'evy time changed jump-diffusion $\widetilde X_t=X_{Z_t}$ is polynomial. But $\widetilde X_t$ is not affine if $\kappa\neq 0$. Indeed, straightforward integration shows
\[     \widetilde\Gcal \e^{ux} =   \int_E  \left(\e^{uy} -\e^{ux} \right) p_1(x,dy) = \left( \e^{\left(\e^{-\kappa t}-1\right)u x +C(t)} -1\right)\e^{ux} \]
for $C(t)=\frac{\sigma^2 u^2}{4\kappa}\left(1-\e^{-2\kappa t}\right)$, which is not of the form~\eqref{eq:D:affine}.
\end{ex}

Applications of L\'evy time changed Ornstein--Uhlenbeck processes as in Example~\ref{exLTCOU} to commodity derivatives pricing are given in \citet{li_lin_14}.

\section{Polynomial Expansions}\label{secpolyexp}

We study the generic pricing problem in finance, which can be cast as follows. Let $X_t$ be polynomial jump-diffusion on state space $E\subseteq \R^d$. Pricing a possibly path-dependent option boils down to computing the conditional expectation
\[ I_{t_0} = \E[ F(\bm X)\mid\Fcal_{t_0}] \]
where $\bm X=P(X_{t_1},\dots,X_{t_n})$ for some linear map $P:\R^{d\times n}\to \R^{m}$, with $m\le d\times n$, for some time partition $0\le t_0<t_1<\cdots<t_n$, and some discounted payoff function $F(\bm x)$ on $\R^m$. For example, $\bm X=(X_{1,t_1},\dots,X_{1,t_n})$ may only depend on the first component of $X_t$, so that $m=n$. In the following we present a method that extends the approach in \cite{fil_may_sch_13}.

We denote by $g(\bm {dx})$ the regular conditional distribution of $\bm X$ on $\R^m$ given $\Fcal_{t_0}$. We let $w(\bm {d x})$ be an auxiliary probability kernel from $(\Omega,\Fcal_{t_0})$ to $\R^m$ that dominates $g(\bm{dx})$, with likelihood ratio function denoted by $\ell(\bm x)$, such that
 \begin{equation}\label{assPE0}
  g(\bm{dx})=\ell(\bm x)w(\bm{dx}).
 \end{equation}
We define the Hilbert space $L^2_w$ as the set of (equivalence classes of) measurable real functions $f(\bm x)$ on $\R^m$ with finite $L^2_w$-norm given by
\[ \left\| f\right\|_w^2 =\int_{\R^m}  f(\bm x)^2 w(\bm{dx}) .\]
The corresponding scalar product is $\left\langle f,h\right\rangle_w = \int_{\R^m}  f(\bm x)h(\bm x) w(\bm{dx})$. We assume that $L^2_w$ contains all polynomials on $\R^m$,
\begin{equation}\label{assPE1}
 \Pol(\R^m)\subset L^2_w,
\end{equation}
and let $q_0(\bm x)=1, q_1(\bm x) , q_2(\bm x) ,\dots$ form an orthonormal basis of polynomials spanning the closure $\overline{\Pol(\R^m)}$ in $L^2_w$. We also assume that the likelihood ratio function lies in $L^2_w$,
\begin{equation}\label{assPE2}
  \ell \in L^2_w .
\end{equation}
As a consequence, its Fourier coefficients
\begin{equation}\label{ellkeq}
  \ell_k =\langle \ell,q_k \rangle_w = \int_{\R^m}  q_k(\bm x)\ell(\bm x) w(\bm{dx}) = \E\left[ q_k(\bm X)\mid\Fcal_{t_0}\right]
\end{equation}
are given in closed form by iterating the moment formula in Theorem~\ref{T:moments}.

We finally assume that the discounted payoff function lies in $L^2_w$,
\begin{equation}\label{assPE3}
  F \in L^2_w .
\end{equation}
We denote by $\bar F$ the orthogonal projection of $F$ onto $\overline{\Pol(\R^m)}$ in $L^2_w$. Elementary functional analysis then gives that the price approximation $\bar I_{t_0} = \E[ \bar F(\bm X)\mid\Fcal_{t_0}] $ equals
\begin{equation}\label{sereq}
 \bar I_{t_0} = \int_{\R^m}  \bar F(\bm x) g(\bm{dx}) =\left\langle \bar F,\ell\right\rangle_w = \sum_{k\ge 0}   F_k \ell_k
\end{equation}
with Fourier coefficients given by
\begin{equation}\label{Fkeq}
   F_k = \langle \bar F,q_k \rangle_w = \langle F,q_k  \rangle_w =\int_{\R^m} F(\bm x)  q_k(\bm x) w(\bm{dx}).
\end{equation}

The approximation equals the true price, $\bar I_{t_0}=I_{t_0}$, if the projection $\bar F$ equals $F$ in $L^2_w$. This statement is more of theoretical than of practical interest for two reasons. First, depending on the choice of the auxiliary kernel $w(\bm{dx})$, we have that $\Pol(\R^m)$ is dense in $L^2_w$, such that $\bar F=F$ holds anyway. Second, in practice we approximate the price by truncating the series in \eqref{sereq} at some finite order~$K$,
\begin{equation}\label{IKeq}
 I_{t_0}^{(K)} = \sum_{k=0}^K   F_k \ell_k ,
\end{equation}
such that the pricing error is $\epsilon^{(K)} = I_{t_0}-I_{t_0}^{(K)}$. While it is good to know that $\epsilon^{(K)}\to 0$ asymptotically as $K\to\infty $ if $\bar F=F$ in $L^2_w$, the hard work consists in controlling the error $\epsilon^{(K)}$ for finite $K$.

The computation of the approximation $I_{t_0}^{(K)}$ can be casted as numerical integration over $\R^m$,
\begin{equation}\label{Feq1}
  I_{t_0}^{(K)} = \sum_{k=0}^K   \langle F, \ell_k q_k\rangle_w  = \int_{\R^m} F(\bm x) \ell^{(K)}(\bm x) w(\bm{dx}) ,
\end{equation}
for the likelihood ratio approximation
\[ \ell^{(K)}(\bm x) =  \sum_{k=0}^K \ell_k q_k(\bm x). \]
Note that the approximation $g^{(K)}(\bm{dx}) =\ell^{(K)}(\bm x)   w(\bm{dx}) $ of the measure $g(\bm{dx})$ integrates to one, $g^{(K)}({{\R^m}})=1$, because $q_k$ is orthogonal to $q_0=1$ in $L^2_w$ for $k\ge 1$. But $g^{(K)}(\bm{dx})$ is only a signed measure in general.

How to choose the auxiliary probability kernel $w(\bm{dx})$? Necessarily $w(\bm{dx})$ has to satisfy conditions~\eqref{assPE0}--\eqref{assPE2} and \eqref{assPE3}, whereof \eqref{assPE2} is arguably the most difficult to verify in practice.\footnote{In Section~\ref{applacmc} we sketch a situation that one may encounter in applications.} The following criteria indicate desirable further properties of $w(\bm{dx})$ from a computational point of view:

\begin{enumerate}
  \item\label{desP1} $w(\bm{dx})$ admits closed-form $\Fcal_{t_0}$-conditional moments. Then we obtain the orthonormal polynomials $q_0(\bm x)=1, q_1(\bm x) , q_2(\bm x) ,\dots$  in $L^2_w$ in closed-form and without numerical integration. Indeed, we let $\tilde q_0(\bm x)=1, \tilde q_1(\bm x),\tilde q_2(\bm x),\dots$ be any basis of $\Pol(\R^m)$. We obtain all scalar products $\langle \tilde q_k,\tilde q_l\rangle_w$ in terms of the $\Fcal_{t_0}$-conditional moments of $w(\bm{dx})$. This allows to perform an exact Gram--Schmidt orthonormalization and we obtain an orthonormal basis of $\overline{\Pol(\R^m)}$ in $L^2_w$ in closed-form.

\item\label{desP2} There exist closed-form formulas for the Fourier coefficients $F_k$ and no numerical integration is needed for the computation of $I^{(K)}_{t_0}$. See for example the option pricing in \cite{ack_fil_pul_16, ack_fil_17}. Otherwise, one has to numerically integrate \eqref{Fkeq} or, equivalently, \eqref{Feq1} with respect to $w(\bm{dx})$. This should then at least be amendable by cubature or Monte--Carlo methods.

\item\label{desP3} $w(\bm{dx})$ matches the moments of $g(\bm{dx})$ of order $n$ and less, $\int_{\R^m} \bm x^{\bm\alpha} w(\bm{dx})=\int_{\R^m} \bm x^{\bm\alpha} g(\bm{dx})$ for all $|\bm\alpha|\le n$. We already know this always holds for $\bm\alpha=\bm 0$, so that $\ell_0=1$. Then $\ell_k=\int_{\R^m}  q_k(\bm x)g(\bm{dx})=\int_{\R^m}  q_k(\bm x)w(\bm{dx})=\langle 1,q_k\rangle_w=0$ for all $k\ge 1$ with $\deg q_k\le n$. This can improve the convergence of the approximation~\eqref{IKeq}. A numerically efficient method for constructing a probability density matching the first $n$ moments in the univariate case, $m=1$, is presented by \citet[Section~3.2]{fil_wil_17}.
\end{enumerate}

\section{Polynomial Asset Pricing Models}\label{secPAPM}

Building on Section~\ref{seccondLevy}, we develop a polynomial framework that accommodates a large class of asset pricing models. It nests all affine asset pricing models, subject to integrability of jumps. First, we introduce the financial market model with excess log returns that are conditional L\'evy based on a polynomial jump-diffusion factor process. Then we discuss option pricing and equivalent measure change, and provide closed form expressions for the return volatility, vol of vol, and leverage.

\subsection{Conditional L\'evy Excess Log Returns}

We consider a financial market with $e$ primary assets with price processes given as
\[
S_{i,t} = S_{i,0}\, \e^{  \int_0^t r_s ds + Y_{i,t}},
\]
where $r_t$ is the risk-free rate and $Y_t=(Y_{1,t},\dots,Y_{e,t})$ are the excess log return processes with $Y_0=0$. We let $X_t$ be a polynomial jump-diffusion on some state space $E\subseteq\R^d$ such that $Z_t=(X_t,Y_t)$ is a an $E\times\R^e$-valued jump-diffusion with extended generator $\Gcal$ of the form~\eqref{eq:GZ_19}. That is, $Y_t$ is a conditional L\'evy process. We follow the conventions~\eqref{aXaY_19}--\eqref{nuXnuY_19} and assume
\begin{equation}\label{assnuYexp}
\int_{\R^e} \left(\e^{\eta_i}-1-\eta_i\right)\nu^Y(x,d\eta)<\infty\quad\text{for all $x\in E$, $i=1,\dots,e$.}
\end{equation}
It then follows that the price processes are special semimartingales with the decomposition
\[ \frac{dS_{i,t}}{S_{i,t-}} = (r_t + \epsilon_i(X_t))\,dt + dY^c_{i,t}+ \int_{\R^e} \left( \e^{\eta_i}-1\right)\left(\mu^{Y}(d\eta,dt)-\nu^Y(X_{t-},d\eta)\right), \]
where $Y^c_t$ denotes the continuous martingale part of $Y_t$, $\mu^Y(d\eta,dt)$ is the integer-valued random measure associated to the jumps of $Y_t$, and the excess rates of return are given by
\[ \epsilon_i(x) = b^Y_i(x) + \frac{1}{2} a^Y_{ii}(x) + \int_{\R^e} \left(\e^{\eta_i}-1-\eta_i\right)\nu^Y(x,d\eta) .\]

We have not specified the measure $\P$ yet. For derivatives pricing, we assume that $\P$ is a risk-neutral measure, so that the discounted price processes $\e^{-\int_0^t r_s ds}S_{i,t}$ are local martingales. This is achieved by setting $\epsilon_i(x)=0$ for all $x\in E$ and $i=1,\dots,e$. In view of Lemma~\ref{LemPPchar}, $\epsilon_i(x)$ cannot be zero if $Z_t=(X_t,Y_t)$ were a polynomial jump-diffusion, other than affine, in general. The next result shows how to embed $Z_t$ into a higher dimensional polynomial jump-diffusion such that $\epsilon_i(x)=0$. It follows immediately from Theorem~\ref{T:augNEW1_19}, so that we omit its proof.

\begin{lem}\label{lemPAPbar}
Let $n\in\N$. Assume that \eqref{assnuYexp} holds and
\begin{align}
a^Y&\in\Pol_n(E),\notag \\
a^{XY}&\in\Pol_{1+n}(E),\notag \\
\int_{\R^{d+e}} \xi^{\bm\alpha} f(\eta)\,\nu(\cdot,d\xi\times d\eta)&\in\Pol_{|\bm\alpha|+n}(E),\quad\text{for all $|\alpha|\ge 0$,} \label{nuxif}
\end{align}
for all functions $f:\R^e\to\R$ for which the integral is finite. Then one can choose
\[ b_i^Y =- \frac{1}{2} a^Y_{ii}  - \int_{\R^e} \left(\e^{\eta_i}-1-\eta_i\right)\nu^Y(\cdot,d\eta) \in \Pol_n(E) ,\]
so that $\epsilon_i(x)=0$ on $E$ and $\P$ is a risk-neutral measure, and $\overline Z_t=(H(X_t),Y_t)$ is a polynomial jump-diffusion on $H(E)\times\R^e$.
\end{lem}

\subsection{Option Pricing}
To illustrate how to price options on the primary assets $S_{i,t}$ we now assume that $\P$ is a risk-neutral measure. Consider first a European call option written on asset $S_{i,t}$ with strike $K$ and maturity $T$. Its price at time $t=0$ is given by
\[
\E\left[ \e^{-\int_0^T r_sds} (S_{i,T} - K)^+ \mid\Fcal_0\right] = \E\left[  (S_{i,0}\e^{Y_{i,T}}- K\e^{-\int_0^T r_sds} )^+ \mid\Fcal_0\right].
\]
If the risk-free rates $r_t$ are deterministic, the pricing operation reduces to computing expectations of the form $\E[  (\e^{Y_{i,T}} - c )^+\mid\Fcal_0 ]$, where $c$ is a constant. Pricing path-dependent derivatives boils down to computing conditional expectations of the form $\E [ F(Y_{i,t_1},\dots,Y_{i,t_n})\mid\Fcal_0  ]$, as discussed in Section~\ref{secpolyexp} for $\overline Z_t$ in lieu of $X_t$.

\subsection{Equivalent Measure Change}\label{S_measure_change}
In general, $\P$ may be any measure, such as the real-world measure, as long as there exists a locally equivalent risk-neutral measure $\Q$ such that the discounted price processes are $\Q$-local martingales. This is a standard condition on asset pricing models to be arbitrage-free, see~\cite{har_pli_81}. In case $\P$ is not a risk-neutral measure, we specify the market price of risk such that $X_t$ is a polynomial-jump diffusion under the corresponding risk-neutral measure. Thereto we fix a finite time horizon $T$ and consider equivalent probability measures $\Q\sim\P$ under which $Z_t$, $t\in [0,T]$, is a jump-diffusion with diffusion, drift, and jump coefficients $a^\Q(x)$, $b^\Q(x)$, and $\nu^\Q(x,d\zeta)$ given in terms of the $\P$-coefficients $a(x)$, $b(x)$, and $\nu(x,d\zeta)$ as
\begin{equation}\label{Qchars}
 \begin{aligned}
a(x) &= a^\Q(x),   \\
b(x) &= b^\Q(x) + a(x)\phi(x) + \int_{\R^{d+e}} \left(1-1/\psi(x,\zeta)\right) \zeta\, \nu(x,d\zeta),   \\
\nu(x,d\zeta ) &= \psi(x,\zeta)\,\nu^\Q(x,d\zeta ),
\end{aligned}
\end{equation}
for some $\R^{d+e}$-valued function $\phi(x)$ and real function $\psi(x,\zeta)>0$. Here $\phi(x)$ is the market price of diffusion risk and $\psi(x,\zeta)$ is the market price of risk of the jump event of size $\zeta$ associated with $Z_t$.\footnote{If the continuous martingale part of $Z_t$ is of the form $dZ^c_t = \sigma(X_t)\,dW_t$, for some Brownian motion $W_t$, and $a(x)=\sigma(x)\sigma(x)^\top$, then $\sigma(X_t)^\top\phi(X_t)$ is the market price of risk of $W_t$.}
$\Q$ is a risk-neutral measure if $\epsilon^\Q_i(x)=0$ for all $x\in E$ and $i=1,\dots,e$.

In order that the change of measure is well defined, we assume that
\begin{equation}\label{assQ1}
  \text{$\Ecal_t(L)$, $t\in [0,T]$, is a positive martingale}
\end{equation}
for
\[
dL_t = - \phi(X_t)^\top d Z^c_t - \int_{\R^{d+e}} \left(1-1/\psi(X_{t-},\zeta)\right) \left(\mu^{Z}(d\zeta,dt) - \nu(X_{t-},d\zeta)dt\right)
\]
with $L_0=0$, where $Z^c_t$ is the continuous martingale part of $Z_t$ and $\mu^{Z}(d\zeta,dt)$ denotes the integer-valued random measure associated to the jumps of $Z_t$.\footnote{This assumption entails that $\phi(x)$ and $\psi(x,\zeta)>0$ are measurable and such that
\[ \int_0^T \phi(X_t)^\top a(X_t)\phi(X_t)\,dt + \int_0^T \int_{\R^{d+e}} \left(1-\sqrt{1/\psi(X_t,\zeta)} \right)^2  \nu(X_t,d\zeta)dt  <\infty, \]
so that $L_t$, $t\in [0,T]$, is a well defined local martingale, see \citet[Theorem II.1.33d]{Jacod/Shiryaev:2003}.}
Then $ d\Q / d\P=\Ecal_T(L)$ defines an equivalent probability measure $\Q\sim\P$. We also assume that
\begin{equation}\label{assQ2}
 \int_0^T \int_{\R^{d+e}} (\|\zeta\|^2 \wedge \|\zeta\| ) /\psi(X_t,\zeta)\,\nu(X_t,d\zeta)dt <\infty.
\end{equation}
Girsanov's Theorem then implies that $Z_t$, $t\in [0,T]$, is a jump-diffusion with diffusion, drift, and jump coefficients $a^\Q(x)$, $b^\Q(x)$, and $\nu^\Q(x,d\zeta)$ given in \eqref{Qchars}, see \citet[Theorem III.3.24]{Jacod/Shiryaev:2003}.

In addition to \eqref{assQ1} and \eqref{assQ2}, assume $\int_{\R^{d+e}} \|\zeta\|^{k}/\psi(x,\zeta)\,\nu(x,d\zeta)<\infty$ for all $x\in E$ and all $k\ge 2$. Let $\Gcal^\Q$ denote the generator of $Z_t$ under $\Q$, so that
\[ \Gcal^\Q f(z) = \Gcal f(z) -\phi(x)^\top a(x)\nabla f(z) - \int_{\R^{d+e}} (f(z+\zeta)-f(z) \left(1-1/\psi(x,\zeta)\right)  \nu(x,d\zeta).\]
Since $\Gcal$ is well-defined on $\Pol(E\times\R^e)$, Lemma~\ref{lemGamma0} implies that $\Gcal^\Q$ is well-defined on $\Pol(E\times\R^e)$. On a case-by-case basis it is now straightforward to derive conditions from \eqref{Qchars} and Lemma~\ref{LemPPchar} such that $X_t$, $t\in [0,T]$, is a polynomial jump-diffusion under $\Q$ and such that Lemma~\ref{lemPAPbar} applies; so that $\Q$ is a risk-neutral measure and $\overline Z_t=(H(X_t),Y_t)$, $t\in[0,T]$, is a polynomial jump-diffusion on $H(E)\times\R^e$ under $\Q$.

\subsection{Volatility, Vol of Vol, and Leverage}

The spot variance $v_i(X_{t-})$ of the $i$th excess log return $dY_{i,t}$ is defined as the time derivative of the predictable compensator of its quadratic variation $[Y_i,Y_i]_t$ (modified second characteristic) given by
\[ v_i(x) = \Gamma(y_i,y_i)(x)=a^Y_{ii}(x) + \int_{\R^e} \eta_i^2 \nu^Y(x,d\eta) \]
where $\Gamma$ denotes the carr\'e-du-champ operator related to $\Gcal$ (see Section~\ref{S:cdc}).

The volatility of $dY_{i,t}$ is defined as the square-root of its spot variance, ${\rm vol}_i(X_{t-})=\sqrt{v_i(X_{t-})}$. The vol of vol is defined as the square-root of the spot variance of the volatility process ${\rm vol}_i(X_t)$,
\[ {\rm volvol}_i(X_{t-}) = \sqrt{\Gamma ({\rm vol}_i,{\rm vol}_i)(X_{t-})}.\]

The leverage effect refers to the generally negative correlation between $dY_{i,t}$ and changes of its spot variance $dv_i(X_t)$. It is captured by the time derivative of the predictable compensator of the quadratic co-variation between $Y_{i,t}$ and $v_i(X_t)$,
\[  {\rm lev}_i(X_{t-}) = \frac{\Gamma(y_i,v_i)(X_{t-})}{\sqrt{v_i(X_{t-})}\sqrt{\Gamma(v_i,v_i)(X_{t-})}} .\]
Note that in the presence of jumps, the jump measure $\nu^Y(x,d\eta)$ and thus the spot variance, volatility, vol of vol, and the leverage depend on the measure $\P$, so that we distinguish risk-neutral and real-world volatility, vol of vol, and leverage.\footnote{Some authors restrict to the diffusive component of $dY_{i,t}$ for the definitions of spot variance, volatility, vol of vol, and leverage, which are the same under both measures.}

\section{Linear Volatility}\label{secLVM}

We introduce a large class of polynomial asset pricing models based on the linear SDE~\eqref{SDElin}--\eqref{SDElincoeff}, extending Example~\ref{ex_5_1}. Throughout this section we assume that $\P$ is a risk-neutral measure. Let $W_t$ be a standard $m$-dimensional Brownian motion. Let $N(du,dt)$ be a Poisson random measures with compensator $F(du)dt$ on $U\times \R_+$, for some mark space $U$.\footnote{The Poisson random measure $N(du,dt)$ and the Brownian motion $W_t$ are automatically independent, see \citet[Theorem II.6.3]{ike_wat_89}.} We assume that $X_t$ is the $E$-valued solution of a linear SDE
\begin{equation}\label{dXeqLV}
 dX_t = b^X(X_t)\,dt + \sigma^X(X_t)\,dW_t + \int_U \delta^X(X_{t-},u) (N(du,dt)-F(du)dt),
\end{equation}
where drift, volatility, and jump size functions $b^X(x)$, $\sigma^X(x)$, and $\delta^X(x,u)$ are linear in $x$ of the form \eqref{SDElincoeff}, for some state space $E\subseteq\R^d$. We then specify the excess log returns by
\begin{equation}\label{dReqLV}
  dY_t = b^Y(X_t)\,dt+ \sigma^Y(X_t)\,dW_t +   \int_{U} \delta^Y(u) (N(du,dt)-F(du)dt) ,
 \end{equation}
with drift $b^Y(x)$ to be determined such that $\P$ is a risk-neutral measure. The volatility function is linear,
\[   \sigma^Y(x)= \Gamma^Y_0 +\sum_{i=1}^d x_i\Gamma^Y_i,\]
for parameters $\Gamma^Y_i\in\R^{e\times m}$, $i=0,\dots,d$. Jumps of $Y_t$ are captured by the state-independent jump size function $\delta^Y(u)$ and can be isolated or simultaneous with jumps of $X_t$.\footnote{State-dependent jumps of $Y_t$ that are simultaneous with jumps of $X_t$ would violate the structural condition~\eqref{nuxif}.}
We assume that the pushforward $\delta^Y_\ast F(d\eta)$ of $F(du)$ under $\delta^Y(u)$ satisfies \eqref{assnuYexp},
\begin{equation}\label{assdeltaYF}
\int_U  \left(\e^{\delta^Y_i(u)}-1-\delta^Y_i(u)\right)F(du)<\infty\quad\text{for all $i=1,\dots,e$.}
\end{equation}

The resulting coefficients $a(x)$, $b(x)$, and $\nu(x,d\xi\times d\eta)$ of the generator of the jump-diffusion $Z_t=(X_t,Y_t)$ are functions of $x$ given by
\[ a=\begin{pmatrix}
 \sigma^X\sigma^{X \top} & \sigma^X\sigma^{Y \top} \\ \sigma^Y\sigma^{X \top} & \sigma^Y\sigma^{Y \top}
\end{pmatrix}\in\Pol_2(E),\quad b=\begin{pmatrix}b^X \\ b^Y\end{pmatrix}, \]
with $b^X\in\Pol_1(E)$, and
\[
 \int_{\R^{d+e}} f(\xi,\eta)\,\nu(x,d\xi\times d\eta) =   \int_{U} f(\delta^X(x,u),\delta^Y(u)) \,F(du).\]
It follows by inspection that the assumptions of Lemma~\ref{lemPAPbar} are met for $n=2$. Hence we can set
\[ b_i^Y =- \frac{1}{2} a^Y_{ii}  - \int_U  \left(\e^{\delta^Y_i(u)}-1-\delta^Y_i(u)\right)F(du) \in \Pol_2(E) ,\]
so that $\P$ is a risk-neutral measure, as desired, and $Z_t=(X_t,Y_t)$ satisfies properties~{\eqref{DefPolaug1_19}}--{\eqref{eq:T:aug:nuew_19}} in Theorem~\ref{T:augNEW1_19} for $n=2$.

Here is an example for the specification of the $N(du,dt)$-driven jumps of $(X_t,Y_t)$.
\begin{ex}\label{exJump1}
Let $U=U_0\cup   U_1\cup\cdots\cup U_k$ for pairwise disjoint sets $U_j$ such that $\lambda_j=F(U_j)<\infty$. Then $N_{j,t}=N(U_j\times [0,t))$ are independent Poisson processes with intensities $\lambda_j$, for $j=0,\dots,k$. Define the piece-wise constant $\delta^X(x,u)=0$ for $u\in U_0$ and $\delta^X(x,u)=\delta^X_{0}+\sum_{i=1}^d x_i\delta^X_{ij}$ for $u\in U_j$, $j=1,\dots,k$, for some parameters $\delta^X_{0}, \delta^X_{ij}\in\R^d$. Then the $N(du,dt)$-driven jump term in \eqref{dXeqLV} reads
\[ \int_U \delta^X(X_{t-},u) (N(du,dt)-F(du)dt) = \sum_{j=1}^k\left(\delta_0+\sum_{i=1}^d X_{i,t-}\delta^X_{ij}\right) (dN_{j,t}-\lambda_j\,dt) .\]
The $N(du,dt)$-driven jump term in \eqref{dReqLV} accordingly is of the form
\[ \int_0^t\int_{U} \delta^Y(u) (N(du,ds)-F(du)ds) = \sum_{j=0}^k \left( \sum_{i\in \N} {H}^{(i)}_{j}1_{\left\{ i\le N_{j,t} \right\}} -\E[{H}^{(1)}_{j}]\lambda_j t\right) \]
where $\delta^Y(u)$ is any jump size function on $U$ satisfying~\eqref{assdeltaYF}, ${H}^{(i)}_j$ are mutually independent $\R^e$-valued random variables, independent of the Poisson processes $N_{0,t},\dots,N_{k,t}$ and the Brownian motion $W_t$, and such that ${H}^{(i)}_j$ has distribution given by the pushforward of $F(u)/\lambda_j$ under the restriction $\delta^Y|_{U_j}$, that is, $\E[f({H}^{(i)}_j)]=\int_{U_j} f(\delta^Y(u)) F(du)/\lambda_j$, for $i\in\N$ and $j=0,\dots,k$. Note that $N_{0,t}$ drives isolated jumps of $Y_t$. On the other hand, $N_{j,t}$ drives isolated jumps of $X_t$ if $H_j^{(i)}=0$, for some $j=1,\dots,k$
\end{ex}

Examples of linear diffusion volatility models include the extended Stein--Stein model \citep{ste_ste_91,sch_zhu_99} and the extended Hull--White model \citep{hul_whi_87,lio_mus_07}, which are discussed in detail in \citet{ack_fil_17}.

\section{Conclusion}\label{secconc}

We have developed a mathematical framework for polynomial jump-diffusions in a relaxed semimartingale context, as opposed to a Markovian setup, based on the point-wise action of the extended generator on polynomials only. We have established various features, including the moment formula and the invariance with respect to polynomial transformations and L\'evy time change. We have also revisited affine jump-diffusions, which are nested in the polynomial class, in that relaxed context. We have then constructed a large class of novel asset pricing models based on polynomial jump-diffusions and presented a generic method for option pricing. Our results provide the basis for new asset pricing models. Several extensions are possible and left for future research. This includes discrete time and time-inhomogeneous polynomial jump-diffusions.

\begin{appendix}

\section{Carr\'e-du-champ Operator} \label{S:cdc}

Let $X_t$ be a jump-diffusion with extended generator $\Gcal$ on $\R^d$ of the form \eqref{eq:G}. An object closely related to $\Gcal$ is the carr\'e-du-champ operator, which is a bilinear operator $\Gamma$ given by
\begin{equation} \label{eq:Gamma}
\Gamma(f,g)(x) = \Gcal(fg)(x) - f(x)\Gcal g(x) - g(x)\Gcal f(x)
\end{equation}
for any $C^2$ functions $f(x)$ and $g(x)$ on $\R^d$ such that
\[ \text{$\int_{\R^d} \left( h(x+\xi) - h(x) \right)^2  \nu(x,d\xi)<\infty$ for all $x\in\R^d$ for $h=f,g$.}\]
Using the product rule, one can express $\Gamma(f,g)$ in terms of $a(x)$ and $\nu(x,d\xi)$ as
\begin{equation} \label{eq:Gamma2}
\Gamma(f,g)(x) = \nabla f(x)^\top a(x) \nabla g(x) + \int_{\R^d} (f(x+\xi) - f(x))(g(x+\xi) - g(x))\nu(x,d\xi).
\end{equation}
In particular, $\Gamma(f,f)\ge 0$. Just as the extended generator captures the drift of a process $f(X_t)$, the carr\'e-du-champ operator gives information about the quadratic variation of $f(X_t)$.

\begin{lem} \label{L:QV Gamma}
The predictable compensator of the quadratic co-variation $[f(X),g(X)]_t$ is given by
\[ \int_0^t \Gamma(f,g)(X_s)\,ds\]
for any $C^2$ functions $f(x)$ and $g(x)$ on $\R^d$ such that
\begin{equation}\label{condfgnu}
 \int_0^t\int_{\R^d} \left|  f(X_s+\xi) - f(X_s)  \right|\left|  g(X_s+\xi) - g(X_s)  \right|\nu(X_s,d\xi)\,ds < \infty.
\end{equation}
\end{lem}

\begin{proof}
In view of \citet[Theorem~I.4.52]{Jacod/Shiryaev:2003} and the fact that $d\langle X^c,X^c\rangle_t=a(X_t)dt$, we have
\[
[f(X),g(X)]_t=\int_0^t \nabla f(X_s)^\top a(X_s) \nabla g(X_s)\,ds + \sum_{s\le t} (f(X_s)-f(X_{s-}))(g(X_s)-g(X_{s-})).
\]
The result now follows from~\eqref{eq:Gamma2} and \citet[Theorem~II.1.8]{Jacod/Shiryaev:2003}.
\end{proof}

With the help of the carr\'e-du-champ operator we can qualify the property of $\Gcal$ being well-defined on $\Pol(E)$ for some state space $E\subseteq\R^d$.

\begin{lem}\label{lemGamma0}
Assume $\Gcal$ is well-defined on $\Pol(E)$. Let $f\in\Pol(\R^d)$ with $f(x)=0$ on $E$. Then $a(x)\nabla f(x) = 0$ and $\int_{\R^d} (f(x+\xi) - f(x))^2 \nu(x,d\xi)=0$ on $E$.
\end{lem}

\begin{proof}
We have $f(x)^2=0$ on $E$, so that \eqref{eq:Gamma} implies $\Gamma(f,f)(x)=0$ on $E$. Since $a(x)\in \S^d_+$, the lemma follows from identity~\eqref{eq:Gamma2}.
\end{proof}

\section{Polynomial Transformations of Jump-Diffusions}

We let $X_t$ be an $E$-valued jump-diffusion with extended generator $\Gcal$ of the form~\eqref{eq:G}, for some state space $E\subseteq\R^d$. We show that an invertible polynomial transformation of $X_t$ is again a jump-diffusion, subject to technical conditions, and we identify its extended generator.

\begin{lem} \label{L:JD transformNEW_2}
Let $\varphi:\R^d\to\R^k$ be a polynomial map that admits a measurable inverse on $E$, in the sense that there exists a measurable map $\psi:\R^k\to\R^d$ such that $\psi\circ\varphi=\id$ on $E$. Assume that $\Gcal$ is well-defined on $\Pol(E)$, i.e.~\eqref{DefPol1} and \eqref{DefPol2} hold. Assume also that the process $\overline X_t=\varphi(X_t)$ is a special semimartingale. Then $\overline X_t$ is a jump-diffusion with extended generator $\overline\Gcal=\psi^*\Gcal \varphi^*$, which is well-defined on $\Pol(\varphi(E))$.\footnote{The pullback is defined in \eqref{eq:pullbackNEW}.} It is of the form
\[
\overline\Gcal f(\overline x) = \frac{1}{2}\tr(\overline a(\overline x)\nabla^2 f(\overline x)) + \overline b(\overline x)^\top \nabla f(\overline x) + \int_{\R^{k}} \left( f(\overline x+\overline \xi) - f(\overline x) - \overline \xi^\top \nabla f(\overline x)\right)\overline \nu(\overline x,d\overline\xi),\]
where, writing $\overline x=\varphi(x)$ and $\varphi_i(x)$ for the $i$th component of $\varphi(x)$,
\begin{align*}
\overline a_{ij}(\overline x) &= \nabla \varphi_i(x)a(x)\nabla\varphi_j(x)^\top,\\
\overline b_i(\overline x) &= \Gcal\varphi_i(x),\\
\overline\nu(\overline x,A) &= \int_{\R^{k}} \bm 1_A(\varphi(x+\xi)-\varphi(x))\nu(x,d\xi).
\end{align*}
\end{lem}

\begin{proof}
Since $\overline X_t$ is a special semimartingale, \citet[Proposition~3]{Kallsen:2006} in conjunction with a direct calculation shows that it is a jump-diffusion with extended generator $\overline\Gcal=\psi^*\Gcal \varphi^*$ of the stated form. In particular, the jump measure satisfies
\[
\int_{\R^k} \|\overline\xi\|^n\, \overline\nu(\overline x,d\overline\xi) = \int_{\R^d} \|\varphi(x+\xi) - \varphi(x)\|^n\, \nu(x,d\xi), \quad \text{where $\overline x=\varphi(x)$,}
\]
which is finite for every $\overline x\in\varphi(E)$ due to \eqref{DefPol1} and since $\varphi(x)$ is polynomial. Finally, if $f\in\Pol(\R^k)$ vanishes on $\varphi(E)$, then $\varphi^*f\in\Pol(\R^d)$ vanishes on $E$, therefore $\Gcal\varphi^*f$ also vanishes on $E$ in view of \eqref{DefPol2}, and hence $\overline\Gcal f$ vanishes on $\varphi(E)$. Thus $\overline\Gcal$ is well-defined on $\Pol(\varphi(E))$.
\end{proof}

\section{Locally Absolutely-Continuous Measure Change}\label{applacmc}
We sketch a situation that may occur in applications for the choice of an auxiliary probability kernel $w(\bm{dx})$ satisfying assumption~\eqref{assPE2}.

Let $\Q$ be a probability measure that is equivalent to $\P$ on each $\Fcal_t$ with Radon-Nikodym density $D_t$. We define $w(\bm{dx})$ as the $\Q$-regular conditional distribution of $\bm X$ given $\Fcal_{t_0}$. Then \eqref{assPE0} holds with likelihood ratio function given by the $\Q$-regular conditional distribution of ${D_{t_m}}/{D_{t_0}}$ given $\Fcal_{t_0}\vee\sigma(\bm X)$,
\begin{equation}\label{claimD1}
  \ell(\bm x)=  \E_\Q\left[\frac{D_{t_m}}{D_{t_0}}\mid \Fcal_{t_0},\,\bm X=\bm x\right] ,
\end{equation}
where we set ${D_{t_m}}/{D_{t_0}}=0$ if $D_{t_0}=0$. Indeed, let $f(\bm x)$ be a bounded measurable function on $\R^m$. Taking conditional expectation gives
\begin{align*}
  \int_{\R^m} f(\bm x)g(\bm{dx}) &= \E_\P\left[ f(\bm X)\mid\Fcal_{t_0}\right]  = \E_\Q\left[ f(\bm X)\frac{D_{t_m}}{D_{t_0}}\mid\Fcal_{t_0}\right]\\
  &= \E_\Q\left[ \E_\Q\left[ f(\bm X)\frac{D_{t_m}}{D_{t_0}}\mid\Fcal_{t_0}\vee\sigma(\bm X)\right] \mid\Fcal_{t_0}\right] \\
  &= \E_\Q\left[ f(\bm X)\ell(\bm X)\mid\Fcal_{t_0}\right] =\int_{\R^m} f(\bm x)\ell(\bm x)w(\bm{dx}),
\end{align*}
which proves the claim~\eqref{claimD1}.

The likelihood ratio function satisfies the estimate
\begin{equation}\label{boundD1}
  \int_{\R^m} \ell(\bm x)^2 w(\bm{dx})= \E_\Q\left[ \E_\Q\left[ \frac{D_{t_m}}{D_{t_0}}\mid\Fcal_{t_0}\vee\sigma(\bm X)\right] ^2 \mid\Fcal_{t_0}\right] \le \E_\Q\left[ \left(\frac{D_{t_m}}{D_{t_0}}\right)^2\mid\Fcal_{t_0}\right].
\end{equation}
The bound in \eqref{boundD1} is sharp to the extent that ${D_{t_m}}/{D_{t_0}}$ could be $\Fcal_{t_0}\vee\sigma(\bm X)$-measurable such that we have equality in \eqref{boundD1}. The estimate \eqref{boundD1} can be useful for verifying assumption~\eqref{assPE2}.

In practice we could approximate $w(\bm{dx})$, the $\Q$-regular conditional distribution of $\bm X$ given $\Fcal_{t_0}$, by simulating $X_t$ under $\Q$. Specifically, we would estimate the Fourier coefficients
\begin{equation}\label{FkQeq}
  F_k = \E_\Q\left[ q_k(\bm X)F(\bm X)\mid\Fcal_{t_0}\right]
\end{equation}
in \eqref{Fkeq} by (nested) Monte-Carlo methods. This addresses property~\ref{desP2}. If we further assume that $X_t$ is a polynomial jump-diffusion with respect to $\Q$, then $w(\bm{dx})$ admits closed-form $\Fcal_{t_0}$-conditional moments, as indicated in property~\ref{desP1}.

\section{Proofs}

This appendix contains the proofs of the lemmas and theorems in the main text.

\subsection{Proof of Theorem~\ref{T:moments}} \label{S:moments_proof}

The proof of Theorem~\ref{T:moments} builds on the following four lemmas.

\begin{lem}\label{lemlocmart}
The local martingale property~\eqref{eq:f-Gf loc mg_NEW} holds for any $C^2$ function $f(x)$ on $\R^d$ satisfying
\begin{equation}\label{eqfff}
 V_t=\int_0^t\int_{\R^d} \left| f(X_s+\xi) - f(X_s) - \xi^\top \nabla f(X_s)\right|\nu(X_s,d\xi)\,ds < \infty.
\end{equation}
\end{lem}

\begin{proof}
Property \eqref{eqfff} states that $V_t$ is in $\Acal^+_{loc}$. The lemma now follows from \citet[Theorem~II.1.8 and proof of Theorem~II.2.42]{Jacod/Shiryaev:2003}.
\end{proof}

For the rest of this section, we assume that $\Gcal$ is polynomial on $E$ and we let $f\in\Pol_n(E)$. Then the process
\[ M^f_t=f(X_t)-f(X_0)-\int_0^t \Gcal f(X_s)\,ds\]
is well defined.

\begin{lem}\label{lemLM}
$M^f_t$ is a local martingale.
\end{lem}

\begin{proof}
In view of Lemma~\ref{lemlocmart} it is enough to show that \eqref{eqfff} holds. But $W(x,\xi)=f(x +\xi) - f(x) - \xi^\top \nabla f(x)$ is a linear combination of monomials $x^{\bm \beta}\xi^{\bm \gamma}$ with $2\le |\bm\gamma|\le n$. Hence $|W(x,\xi)| \le C(x) \left( \|\xi\|^2 + \|\xi\|^{2n}\right)$ for some polynomial $C(x)$. Now \eqref{eqfff} follows from Lemma~\ref{LemPPchar}.
 \end{proof}

\begin{lem} \label{L:moment bound}
For any $k\in\N$ there is a finite constant $C$ such that
\[
\E[ 1 + \|X_t\|^{2k} \mid \Fcal_0] \le \big( 1 + \|X_0\|^{2k} \big)\,e^{Ct}, \quad t\ge0.
\]
\end{lem}

\begin{proof}
We recall the argument in \citet[Theorem~2.10]{Cuchiero/etal:2012} or \citet[Lemma~B.1]{fil_lar_16}. Let $f(x)=1+\|x\|^{2k}$, and let $C$ be a finite constant such that $|\Gcal f(x)|\le Cf(x)$ on $E$. Such a constant exists by the polynomial property of $\Gcal$. Let $0\le T_1\le T_2\le\cdots$ be a localizing sequence for the local martingale $M^f_t$, see Lemma~\ref{lemLM}, such that $\|X_t\|\le m$ for $t<T_m$. Then
\[
\E\left[ f(X_{t\wedge T_m}) \mid \Fcal_0 \right] = f(X_0) + \E\left[ \int_0^{t\wedge T_m} \Gcal f(X_s)\,ds \mid \Fcal_0 \right]
\le f(X_0) + C \int_0^t \E\left[ f(X_{s\wedge T_m}) \mid \Fcal_0 \right]\,ds.
\]
Gronwall's inequality and Fatou's lemma now yield the result.
\end{proof}

\begin{lem} \label{L:sqr int mg}
For any finite $c$ the process $M^f_t 1_{\{\|X_0\|\le c\}}$ is a martingale.
\end{lem}

\begin{proof}
Let $c$ be a finite number. Then $N^f_t=M^f_t 1_{\{\|X_0\|\le c\}}$ is a local martingale by Lemma~\ref{lemLM} with quadratic variation $[N^f,N^f]_t=[f(X),f(X)]_t 1_{\{\|X_0\|\le c\}}$. We claim that its predictable compensator is given by
\[\langle N^f,N^f\rangle_t= \int_0^t \Gamma(f,f)(X_s)\,ds 1_{\{\|X_0\|\le c\}}.\]
Indeed, in view of Lemma~\ref{L:QV Gamma} the claim follows as soon as \eqref{condfgnu} holds for $g=f$. But $W(x,\xi)=(f(x +\xi) - f(x))^2$ is a linear combination of monomials $x^{\bm \beta}\xi^{\bm \gamma}$ with $2\le |\bm\gamma|\le 2n$. Hence $|W(x,\xi)| \le C(x) \left( \|\xi\|^2 + \|\xi\|^{2n}\right)$ for some polynomial $C(x)$, and \eqref{condfgnu} follows from Lemma~\ref{LemPPchar}.

By \eqref{eq:Gamma}, $\Gamma(f,f)(x)$ is a polynomial on $E$. Combining this with Lemma~\ref{L:moment bound} we infer that $\E[\langle N^f,N^f\rangle_t]<\infty$ for all $t\ge 0$, and hence $N^f_t$ is a square-integrable martingale.
\end{proof}

We now prove Theorem~\ref{T:moments}. Fix a finite $c$ and $t\ge 0$. By Lemma~\ref{L:sqr int mg}, the row vector valued function $F(T)=\E[(1,H(X_T)^\top)1_{\{\|X_0\|\le c\}}\mid\Fcal_t]$ satisfies for $T\ge t$
\[F(T) = (1,H(X_t)^\top)1_{\{\|X_0\|\le c\}} + \int_t^T \E[\Gcal (1,H^\top)(X_s)1_{\{\|X_0\|\le c\}}\mid\Fcal_t]\,ds = F(t) +\int_t^T F(s) G\,ds .\]
Hence $\E[(1,H(X_T)^\top)\mid\Fcal_t]1_{\{\|X_0\|\le c\}}=F(T)=(1,H(X_t)^\top)\e^{(T-t)G}1_{\{\|X_0\|\le c\}}$. Theorem~\ref{T:moments} now follows by letting $c\uparrow\infty$.

\subsection{Proof of Lemma~\ref{lemaffine}}
We first assume that $0\in E$ and the affine span of $E$ is all of $\R^d$.

Assume $\Gcal$ is affine. Straightforward calculations show that
\[
\Gcal \e^{ u^\top x}  = \left(\frac{1}{2}u^\top a(x) u + b(x)^\top u + \int_{\R^d} \left( e^{u^\top \xi} - 1 - u^\top\xi \right) \nu(x,d\xi) \right) \e^{ u^\top x}
\]
so that, by virtue of the assumed relation~\eqref{eq:D:affine}, we obtain
\begin{equation} \label{eq:Ri det}
F(u) + R(u)^\top x  = \frac{1}{2}u^\top a(x) u + b(x)^\top u
   + \int_{\R^d} \left( e^{u^\top \xi} - 1 - u^\top\xi \right) \nu(x,d\xi) \quad \text{for all $x\in E$, $u\in{\rm i}\R^d$}.
\end{equation}
We claim that $F(u)$ and $R(u)$ are of the form~\eqref{eq:affine Ri}. Since $0\in E$, this clear for $F(u)$, setting $a_0=a(0)$, $b_0=b(0)$, $\nu_0(d\xi)=\nu(0,d\xi)$. Next, since the affine span of $E$ is all of $\R^d$, there exist numbers $\lambda_1,\ldots,\lambda_d$ with $\sum_{k=1}^d \lambda_k=1$ and points $x^1,\ldots,x^d\in E$ such that $\lambda_1x^1+\cdots+\lambda_dx^d=e_1$, the first canonical unit vector. Evaluating both sides of~\eqref{eq:Ri det} at $x=x^k$, multiplying by $\lambda_k$, summing over $k$, and using the form of $F(u)$, it follows that $R_1(u)$ is of the form~\eqref{eq:affine Ri} with
\[
a_1 = \sum_{k=1}^d \lambda_k a(x^k) - a_0, \qquad b_1 = \sum_{k=1}^d \lambda_k b(x^k) - b_0, \qquad \nu_1(d\xi) = \sum_{k=1}^d \lambda_k \nu(x^k,d\xi) - \nu_0(d\xi).
\]
The same argument shows that $R_2(u),\ldots,R_d(u)$ are also of the form~\eqref{eq:affine Ri}.

It remains to prove~\eqref{eq:affine abnu}. Given $F(u)$ and $R(u)$ just obtained, it is clear that taking $a(x)$, $b(x)$, $\nu(x,d\xi)$ as in~\eqref{eq:affine abnu} is consistent with~\eqref{eq:Ri det}. Furthermore, for each fixed $x\in E$, knowing the right-hand side of~\eqref{eq:Ri det} for all $u\in{\rm i}\R^d$ uniquely determines $a(x)$, $b(x)$, $\nu(x,d\xi)$; see \citet[Lemma~II.2.44]{Jacod/Shiryaev:2003}. Thus~\eqref{eq:affine abnu} is in fact the only possibility, completing the proof of the forward direction.

For the converse, assume $a(x)$, $b(x)$, $\nu(x,d\xi)$ are of the form~\eqref{eq:affine abnu}. A calculation then shows that $\Gcal$ satisfies~\eqref{eq:D:affine} with $F(u)$ and $R(u)$ given by~\eqref{eq:affine Ri}, and thus is affine.

In the general case, where either $0\notin E$ or the affine span of $E$ is not $\R^d$, we apply an invertible affine transformation $T:\R^d\to\R^d$ such that $0\in T(E)$ and the affine span of $T(E)$ is $\R^{d'}\times\{0\}$ for some $d'\le d$. In these new coordinates we set the corresponding $a_i$, $b_i$, and $\nu_i(d\xi)$ to zero for $i>d'$, and then transform back by $T^{-1}$.

\subsection{Proof of Theorem~\ref{thmATF}}
Define the function $f(t,x)=\exp(\phi(T-t)+ \psi(T-t)^\top x)$ and the complex-valued process $M_t=f(t,X_t)$. Then \eqref{eq:Re bound} yields $|M_t|\le 1$. Moreover, a calculation using~\eqref{eq:Riccati} yields
\[
\partial_t f(t,x) + \Gcal f(t,x) = 0, \qquad 0\le t\le T, \quad x\in E,
\]
where $\Gcal$ acts on the real and imaginary parts of $f(t,\cdot)$ separately. Thus $M_t$ is a martingale on $[0,T]$ with $M_T=\exp(u^\top X_T)$. The affine transform formula is now just the equality $M_t=\E[M_T\mid\Fcal_t]$.

\subsection{Proof of Lemma~\ref{lemPBHL}}

The proof of Lemma~\ref{lemPBHL} builds on the following lemma.
\begin{lem} \label{L:degkn}
Any polynomial $p\in\Pol_{mn}(E)$ is of the form $p(x)=f(H(x))$ for some $f\in\Pol_m(H(E))$.
\end{lem}

\begin{proof}
It suffices to consider monomials $p(x)=x^{\bm\alpha}$ with $|{\bm\alpha}|\le mn$. It follows by inspection that ${\bm\alpha}={\bm\alpha}_1+\cdots+{\bm\alpha}_k$ for some multi-indices ${\bm\alpha}_i\in\N^d_0$ with $|{\bm\alpha}_i|\le n$. Thus $x^{\bm\alpha_i}=f_i(H(x))$ on $E$ for some linear polynomial $f_i\in\Pol_1(\R^N)$, for each $i$. We deduce that
\[
p(x) = \prod_{i=1}^m x^{{\bm\alpha}_i} = f(H(x)) \quad\text{on $E$},
\]
where $f(\overline x) = \prod_{i=1}^m f_i(\overline x)$ is of degree at most $m$.
\end{proof}

We now prove Lemma~\ref{lemPBHL}. If a function $f\colon \R^N\to\R$ vanishes on $H(E)$, then $H^*f(x)=f(H(x))$ vanishes on $E$. Thus $H^*$ is well-defined as a map from $\Pol(H(E))$ to $\Pol(E)$, and it is linear with inverse $L^*$. It is clear that $H^*$ maps $\Pol_m(H(E))$ to $\Pol_{mn}(E)$ for each $m\in\N$. To see that $L^*$ maps $p\in\Pol_{mn}(E)$ to an element of $\Pol_m(H(E))$, observe that $p(x)=f(H(x))$ for some $f\in\Pol_m(H(E))$ by Lemma~\ref{L:degkn}, so that $L^*p(\overline x)=f(H(L(\overline x)))=f(\overline x)$. This proves Lemma~\ref{lemPBHL}.

\subsection{Proof of Theorem~\ref{T:exp}}

Since $\overline X_t$ is a special semimartingale by Lemma~\ref{lemLM}, Lemma~\ref{L:JD transformNEW_2} implies that $\overline X_t$ is an $H(E)$-valued jump-diffusion with extended generator $\overline\Gcal=L^*\Gcal H^*$, which is well-defined on $\Pol(H(E))$. Lemma~\ref{lemPBHL} implies that $\overline\Gcal$ is polynomial on $H(E)$, and the diagram~\eqref{diagXX} commutes. This completes the proof of Theorem~\ref{T:exp}.

\subsection{Proof of Theorem \ref{T:augNEW1_19}}

The proof of Theorem~\ref{T:augNEW1_19} builds on the following lemma.

\begin{lem} \label{L:GZ well def_19}
Assume \eqref{DefPolaug1_19}. Then the augmented process $\overline Z_t =(H(X_t),Y_t)$ is a jump-diffusion on $H(E)\times\R^e$ with extended generator $\overline\Gcal = \psi^\ast\Gcal \varphi^\ast$, and the operators $\Gcal$ and $\overline\Gcal$ are well-defined on $\Pol(E\times\R^e)$ and $\Pol(H(E)\times\R^e)$, respectively.
\end{lem}

\begin{proof}
We first prove that $\Gcal$ is well-defined on $\Pol(E\times\R^e)$. Due to \eqref{DefPolaug1_19} we only need to verify \eqref{DefPol2}. Let $f(x,y)$ be a polynomial that vanishes on $E\times\R^e$. Collecting the monomials in $y$ yields the representation
\[
f(x,y) = \sum_{\bm\beta} p_{\bm\beta}(x)y^{\bm\beta}
\]
for finitely many polynomials $p_{\bm\beta}(x)$ on $\R^d$. For each fixed $x\in E$, $f(x,y)$ is the zero polynomial on $\R^e$, and hence $p_{\bm\beta}(x)=0$ for all $\bm\beta$. Thus, we may suppose $f(x,y)=p(x)q(y)$ where $p(x)$ vanishes on $E$ and $q(y)=y^{\bm\beta}$. One has
\[
\Gcal (pq)(x,y) = p(x)\Gcal q(x,y) + q(y)\Gcal p(x,y) + \Gamma(p,q)(x,y).
\]
The first term is zero for any $x\in E$. So is the second term, since $\Gcal p(x,y)=\Gcal^X p(x)$ and $\Gcal^X$ is well-defined on $\Pol(E)$. For the third term, note that the carr\'e-du-champ operator is bilinear and positive semidefinite, and hence satisfies the Cauchy-Schwarz inequality. That is,
\[
|\Gamma(p,q)(x,y)|^2 \le \Gamma(p,p)(x,y)\,\Gamma(q,q)(x,y).
\]
But $\Gamma(p,p)(x,y)=\Gamma^X(p,p)(x)=0$ since $\Gcal$ is well-defined on $\Pol(E)$. Thus $\Gcal (pq)(x,y)=0$ on $E\times\R^e$, and we deduce that $\Gcal$ is well-defined on $\Pol(E\times\R^e)$.

Next, since $H(X_t)$ is a special semimartingale by Lemma~\ref{lemLM}, and since $Y_t$ is a special semimartingale by assumption, $\overline Z_t$ is also a special semimartingale. Therefore, since $\Gcal$ is well-defined on $\Pol(E\times\R^e)$, it follows from Lemma~\ref{L:JD transformNEW_2} that $\overline Z_t$ is a jump-diffusion with extended generator $\overline\Gcal=\psi^*\Gcal\varphi^*$, which is well-defined on $\Pol(H(E)\times\R^e)$.
\end{proof}

We now prove Theorem~\ref{T:augNEW1_19}. Due to Lemma~\ref{L:GZ well def_19}, it remains to prove that \eqref{eq:T:aug:bYNew_19}--\eqref{eq:T:aug:nuew_19} together imply \eqref{T_augNEW1_19_PJD}, and, conversely, that \eqref{T_augNEW1_19_PJD} implies \eqref{eq:T:aug:bYNew_19}, \eqref{eq:T:aug:aYYNew_19}, and \eqref{eq:T:aug:nuew_19} for $\bm\alpha=0$. To do this we make use of Theorem~\ref{T:augNEW2_19}.

We first assume that \eqref{eq:T:aug:bYNew_19}--\eqref{eq:T:aug:nuew_19} hold, and prove property {\ref{T:augNEW12_19}} in Theorem~\ref{T:augNEW2_19}. Fix $m\in\N$ and consider any monomial $f(z)=f(x,y)=x^{\bm\alpha} y^{\bm\beta}$ with $|\bm\alpha|\le n(m-|\bm\beta|)$ and $|\bm\beta|\le m$, which then lies in $V_m$. It suffices to show that $\Gcal f$ again lies in $V_m$. Let $\widetilde a^X(x)$, $\widetilde a^{XY}(x)$, and $\widetilde a^Y(x)$ denote the modified second characteristics of $\Gcal$, that is,
\begin{gather*}
\widetilde a^X(x) = a^X(x) + \int_{\R^d} \xi\xi^\top \nu^X(x,d\xi), \qquad \widetilde a^Y(x) = a^Y(x) + \int_{\R^{d+e}} \eta\eta^\top\nu(x,d\xi\times d\eta), \\
\widetilde a^{XY}(x) = a^{XY}(x) + \int_{\R^{d+e}} \xi\eta^\top\nu(x,d\xi\times d\eta).
\end{gather*}
Furthermore, write $\nabla_xf(x,y)$ for the first $d$ components of $\nabla f(x,y)$, and similarly for $\nabla_y f(x,y)$, $\nabla^2_{xx}f(x,y)$, $\nabla^2_{xy}f(x,y)$, and $\nabla^2_{yy}f(x,y)$. We then get
\begin{align}
\Gcal f(x,y) &= \frac12 \tr\left( \widetilde a^X(x) \nabla^2_{xx}f(x,y) \right) + \tr\left( \widetilde a^{XY}(x) \nabla^2_{xy}f(x,y) \right) + \frac12 \tr\left( \widetilde a^Y(x) \nabla^2_{yy}f(x,y) \right) \label{T:aug:GZ 1} \\
&\quad + b^X(x)^\top \nabla_x f(x,y) + b^Y(x)^\top \nabla_y f(x,y) \label{T:aug:GZ 2} \\
&\quad + \int_{\R^{d+e}} \left( f(z+\zeta) - f(z) - \zeta^\top \nabla f(z) - \frac12 \zeta^\top \nabla^2 f(z) \zeta\right) \nu(x,d\zeta). \label{T:aug:GZ 3}
\end{align}

Consider first \eqref{T:aug:GZ 1}. Since $\nabla^2_{xx}f(x,y)=y^{\bm\beta}\nabla^2(x^{\bm\alpha})$ and since $\Gcal^X$ is polynomial, the first term in~\eqref{T:aug:GZ 1} is of degree at most $|\bm\alpha|$ in $x$ and $|\bm\beta|$ in $y$, and thus lies in $V_m$. Next, $\nabla^2_{xy}f(x,y)$ is of degree at most $|\bm\alpha|-1$ in $x$ and $|\bm\beta|-1$ in $y$, which together with \eqref{eq:T:aug:aXYNew_19} implies that the second term in~\eqref{T:aug:GZ 1} is of degree at most $n+1+|\bm\alpha|-1=n+|\bm\alpha|$ in $x$ and $|\bm\beta|-1$ in $y$. This term therefore also lies in $V_m$. Finally, $\nabla^2_{yy}f(x,y)$ is of degree at most $|\bm\alpha|$ in $x$ and $|\bm\beta|-2$ in $y$, which together with \eqref{eq:T:aug:aYYNew_19} implies that the third term in~\eqref{T:aug:GZ 1} is of degree at most $2n+|\bm\alpha|$ in $x$ and $|\bm\beta|-2$ in $y$. Again this yields membership in $V_m$.

Consider now \eqref{T:aug:GZ 2}. Since $\nabla_xf(x,y)=y^{\bm\beta}\nabla(x^{\bm\alpha})$ and since $\Gcal^X$ is polynomial, the first term in~\eqref{T:aug:GZ 2} is of degree at most $|\bm\alpha|$ in $x$ and $|\bm\beta|$ in $y$, and thus lies in $V_m$. Similarly as above, the second term also lies in $V_m$ due to \eqref{eq:T:aug:bYNew_19}.

Consider finally \eqref{T:aug:GZ 3}. It follows from the multi-binomial theorem that the expression in parentheses is a linear combination of monomials $x^{\bm\gamma} y^{\bm\delta} \xi^{\bm\epsilon}\eta^{\bm\upsilon}$ with $|\bm\gamma|+|\bm\epsilon|\le|\bm\alpha|$, $|\bm\delta|+|\bm\upsilon|\le|\bm\beta|$, and $|\bm\epsilon|+|\bm\upsilon|\ge3$. Thus \eqref{T:aug:GZ 3} is a linear combination of expressions of the form
\[
x^{\bm\gamma} y^{\bm\delta} \int_{\R^{d+e}} \xi^{\bm\epsilon}\eta^{\bm\upsilon} \nu(x,d\xi\times d\eta).
\]
Due to \eqref{eq:T:aug:nuew_19}, these expressions are polynomial of degree at most $|\bm\gamma|+|\bm\epsilon|+n|\bm\upsilon|$ in $x$ and $|\bm\delta|$ in $y$. Since $|\bm\gamma|+|\bm\epsilon|+n|\bm\upsilon|+n|\bm\delta| \le |\bm\alpha|+n|\bm\beta|\le nm$, it follows that \eqref{T:aug:GZ 3} lies in $V_m$. This completes the proof of property {\ref{T:augNEW12_19}} in Theorem~\ref{T:augNEW2_19}, showing that \eqref{T_augNEW1_19_PJD} holds.

Conversely, assume that that \eqref{T_augNEW1_19_PJD} holds, so that property {\ref{T:augNEW13_19}} in Theorem~\ref{T:augNEW2_19} holds as well. Since $\Gcal(y^{\bm\beta})\in V_{|\bm\beta|}$, the identity
\begin{align*}
\Gcal(y^{\bm\beta}) &= \frac12 \tr\left( \widetilde a^Y(x) \nabla^2_{yy}(y^{\bm\beta} \right) + b^Y(x)^\top \nabla_y (y^{\bm\beta}) \\
&\quad + \int_{\R^{d+e}} \left( (y+\eta)^{\bm\beta} - y^{\bm\beta} - \eta^\top \nabla (y^{\bm\beta}) - \frac12 \eta^\top \nabla^2 (y^{\bm\beta})\right) \nu(x,d\xi\times d\eta)
\end{align*}
applied with $\bm\beta=e_i$ yields $b^Y_i \in V_1$, which gives \eqref{eq:T:aug:bYNew_19}. Taking $\beta=e_i+e_j$ we similarly obtain $\widetilde a^Y_{ij} \in V_2$, which gives \eqref{eq:T:aug:aYYNew_19}. By considering $|\bm\beta|\ge3$, we obtain \eqref{eq:T:aug:nuew_19} for $\bm\alpha=0$.

\subsection{Proof of Corollary~\ref{corPQdX_19}}

The process $Z_t=(X_t,Y_t)$ is given by
\[
dZ_t=\begin{pmatrix} dX_t \\ P(X_t)\,dt \\ Q(X_{t-})\,dX_t \end{pmatrix} = K(X_{t-}) \,d\begin{pmatrix} t\\ X_t\end{pmatrix}, \quad\text{where}\quad K(x) = \begin{pmatrix} 0 & \id \\ P(x) & 0 \\ 0 & Q(x) \end{pmatrix}.
\]
Its differential characteristics can then be computed using \citet[Proposition~2]{Kallsen:2006}. One finds that they are deterministic functions $a(x)$, $b(x)$, and $\nu(x,d\zeta)$ of $X_t$, where
\begin{align}
a(x) &= K(x)\begin{pmatrix}0 & 0 \\ 0 & a^X(x) \end{pmatrix} K(x)^\top = \begin{pmatrix} a^X(x) & 0 & a^X(x)Q(x)^\top \\ 0&0&0 \\ Q(x)a^X(x) & 0 & Q(x)a^X(x)Q(x)^\top \end{pmatrix}, \label{corPQdX_eqaZ}\\
b(x) &= K(x)\begin{pmatrix}1\\b^X(x)\end{pmatrix}, \nonumber\\
\nu(x,A) &= \int_{\R^d} \bm 1_A\left(K(x)\begin{pmatrix}0\\ \xi\end{pmatrix}\right) \nu^X(x,d\xi). \nonumber
\end{align}
In particular, we have
\[
\int_{\R^{d+e}} \|\zeta\|^n\, \nu(x,d\zeta) = \int_{\R^d} (\|\xi\|^2+\|Q(x)\xi\|^2)^{n/2}\,\nu^X(x,d\xi) \le (1+\|Q(x)\|^2)^{n/2} \int_{\R^d} \|\xi\|^n\,\nu^X(x,d\xi) < \infty
\]
for all $x\in E$ and all $n\ge2$, where $\|Q(x)\|$ denotes the operator norm of $Q(x)$. Thus \eqref{DefPolaug1_19} holds.

Next, \eqref{eq:T:aug:bYNew_19} holds since the components of $b(x)$ are polynomials of degree at most $n$. To verify \eqref{eq:T:aug:aYYNew_19}--\eqref{eq:T:aug:nuew_19}, we first observe the identity
\begin{equation} \label{eq:C:aug1}
\int_{\R^{d+e}} f(\xi)g(\eta', \eta'')\nu(x,d\xi\times d\eta'\times d\eta'') = \int_{\R^d} f(\xi) g(0,Q(x)\xi) \nu^X(x,d\xi),
\end{equation}
where we write $\eta=(\eta',\eta'')$ for a generic vector in $\R^e=\R^{e'+e''}$. Let $f(\xi)=x^{\bm\alpha}$ and $g(\eta)=\eta^{\bm\beta}=\eta'^{\bm\beta'}\eta''^{\bm\beta''}$, where we decompose $\bm\beta=(\bm\beta',\bm\beta'')$ according to the decomposition $\eta=(\eta',\eta'')$. Since $g(0,Q(x)s\xi)=s^{|\bm\beta''|}g(0,Q(x)\xi)$ for any $s\in\R$, it follows that
\[
g(0,Q(x)\xi) = \sum_{\bm\gamma\colon\,|\bm\gamma|=|\bm\beta''|} r_{\bm\gamma}(x)\xi^{\bm\gamma}
\]
for some polynomials $r_{\bm\gamma}(x)$. Since the left-hand side vanishes if $\bm\beta'\ne0$, we take $r_{\bm\gamma}(x)=0$ in this case. Moreover, since the components of $Q(x)$ are of degree at most $n-1$, the degree of $g(0,Q(x)\xi)$, regarded as a polynomial in $x$, is at most $(n-1)\beta''_1+\cdots+(n-1)\beta''_{e''} = (n-1)|\bm\beta''|$. This is therefore also an upper bound on the degrees of the polynomials $r_{\bm\gamma}(x)$.

These observations readily yield \eqref{eq:T:aug:nuew_19}. To see this, we write
\begin{equation} \label{eq:C:aug2}
\int_{\R^{d+e}} \xi^{\bm\alpha}\eta^{\bm\beta}\nu(x,d\xi\times d\eta) = \sum_{\bm\gamma\colon\,|\bm\gamma|=|\bm\beta''|} r_{\bm\gamma}(x)\int_{\R^d} \xi^{\bm\alpha+\bm\gamma}\nu^X(x,dx).
\end{equation}
Either $\bm\beta'\ne0$, in which case the right-hand side of~\eqref{eq:C:aug2} vanishes. Or, $\bm\beta'=0$, hence $|\bm\beta''|=|\bm\beta|$, and thus the right-hand side of~\eqref{eq:C:aug2} is a polynomial of degree at most $(n-1)|\bm\beta| + |\bm\alpha+\bm\gamma|=|\bm\alpha|+n|\bm\beta|$, provided $|\bm\alpha+\bm\gamma|=|\bm\alpha|+|\bm\beta|\ge3$. Consequently, \eqref{eq:T:aug:nuew_19} holds.

Finally, \eqref{eq:T:aug:aYYNew_19} and \eqref{eq:T:aug:aXYNew_19} follow in a similar manner. For $i\in \{1,\ldots,d\}$ and $j\in\{e'+1,\ldots,e'+e''\}$, we apply~\eqref{corPQdX_eqaZ} and~\eqref{eq:C:aug1} with $f(\xi)=\xi_i$ and $g(\eta)=\eta_{j}$ to get
\[
a^{XY}_{ij}(x) + \int_{\R^{d+e}} \xi_i\eta_j \nu(x,d\xi\times d\eta) = \sum_{k=1}^d q_{j-e',k}(x)\left( a_{ik}(x) + \int_{\R^d} \xi_i\xi_k \,\nu^X(x,d\xi) \right),
\]
which is a polynomial of degree at most $n-1+2=1+n$. If instead $j\in\{1,\ldots,e'\}$, then the left-hand side vanishes. It follows that \eqref{eq:T:aug:aXYNew_19} holds. Property~\eqref{eq:T:aug:aYYNew_19} is proved similarly. This completes the proof of the corollary.

\subsection{Proof of Lemma \ref{lemPBHL2_19}}

Similarly as in the proof of Lemma~\ref{lemPBHL}, the pullback maps $\varphi^*$ and $\psi^*$ are well-defined as operators
\[
\Pol(H(E)\times\R^e)\to \Pol(E\times\R^e) \qquad\text{and}\qquad \Pol(E\times\R^e)\to\Pol(H(E)\times\R^e).
\]
Viewed with these domain and range spaces, $\varphi^*$ and $\psi^*$ are each other's inverses. It is clear that $\varphi^*$ is a linear map. We show that it maps $\Pol_m(H(E)\times\R^e)$ to $V_m$, and for this it suffices to consider monomials $f(\overline x,y)=\overline x^{\bm\alpha}y^{\bm\beta}$, where $|{\bm\alpha}|+|{\bm\beta}|\le m$. One then has
\[
\varphi^*f(x,y)= p(x)y^{\bm\beta}, \quad\text{where}\quad  p(x)= \prod_{i=1}^N h_i(x)^{\alpha_i}.
\]
Since $\deg h_i\le n$ for all $i$, and since $|{\bm\alpha}|\le m-|{\bm\beta}|$, it follows that $\deg p \le n(m - |{\bm\beta}|)$, showing that $\varphi^*f \in V_m$ as claimed.

Injectivity of $\varphi^*$ follows from the identity $\psi^*\circ\varphi^*=\id$ on $\Pol_m(H(E)\times\R^e)$. To see that $\varphi^*$ maps $\Pol_m(H(E)\times\R^e)$ surjectively to $V_m$, let $p(x)y^{\bm\beta}$ with $\deg p\le n(m-|\bm\beta|)$ and $|{\bm\beta}|\le m$ be an element of $V_m$. Lemma~\ref{L:degkn} implies that $p(x) = f(H(x))$ for some polynomial $f\in\Pol_{m-|{\bm\beta}|}(H(E))$. Thus $p(x)y^{\bm\beta}=\varphi^* g(x,y)$, where $g(\overline x,y)=f(\overline x)y^{\bm\beta}$ is of degree at most $m-|{\bm\beta}|+|{\bm\beta}|=m$ and thus lies in $\Pol_m(H(E)\times\R^e)$. Since any element of $V_m$ is a linear combination of polynomials $p(x)y^{\bm\beta}$ as above, this proves surjectivity.

\subsection{Proof of Theorem \ref{T:augNEW2_19}}

We now prove Theorem~\ref{T:augNEW2_19}. Due to Lemma~\ref{L:GZ well def_19}, $\overline Z_t$ is a jump-diffusion on $H(E)\times\R^e$ with extended generator $\overline\Gcal = \psi^\ast\Gcal \varphi^\ast$, which is well-defined on $\Pol(H(E)\times\R^e)$. Thus the property \eqref{T_augNEW1_19_PJD} is, by definition, equivalent to
\begin{equation}\label{eq_augNEW2_19_proof_1}
\text{$\overline\Gcal$ maps $\Pol_m(H(E)\times\R^e)$ to itself for each $m\in\N$.}
\end{equation}
The equivalence ${\eqref{eq_augNEW2_19_proof_1}}\Leftrightarrow{\ref{T:augNEW12_19}}$ follows from Lemma~\ref{lemPBHL2_19} and the expression $\overline\Gcal=\psi^*\Gcal\varphi^*$, which show that in the diagram \eqref{diagXX2_19} each horizontal arrow holds if and only the other one does. In particular, the diagram commutes if either condition holds. It remains to prove ${\ref{T:augNEW12_19}}\Leftrightarrow{\ref{T:augNEW13_19}}$. By definition of the carr\'e-du-champ operator, we have the identity
\[
\Gcal(x^{\bm\alpha}y^{\bm\beta}) = y^{\bm\beta}\Gcal(x^{\bm\alpha}) + x^{\bm\alpha}\Gcal(y^{\bm\beta}) + \Gamma(x^{\bm\alpha},y^{\bm\beta}).
\]
Moreover, if $|\bm\alpha|\le n(m-|\bm\beta|)$ and $|\bm\beta|\le m$ then $y^{\bm\beta}\Gcal(x^{\bm\alpha})=y^{\bm\beta}\Gcal^X(x^{\bm\alpha})\in V_m$ since $\Gcal^X$ is polynomial. Hence {\ref{T:augNEW12_19}} is equivalent to
\begin{equation} \label{T:augNEW1_pf_1}
\text{$x^{\bm\alpha}\Gcal(y^{\bm\beta}) + \Gamma(x^{\bm\alpha},y^{\bm\beta})\in V_m$ whenever $|\bm\alpha|\le n(m-|\bm\beta|)$ and $|\bm\beta|\le m$.}
\end{equation}
It suffices to argue that \eqref{T:augNEW1_pf_1} is equivalent to {\ref{T:augNEW13_19}}. To this end, first observe that
\begin{equation} \label{T:augNEW1_pf_2}
\text{$x^{\bm\alpha}\Gcal(y^{\bm\beta})\in V_m$ whenever $|\bm\alpha|\le n(m-|\bm\beta|)$, $|\bm\beta|\le m$, and $\Gcal(y^{\bm\beta})\in V_{|\bm\beta|}$.}
\end{equation}
Indeed, $\Gcal(y^{\bm\beta})\in V_{|\bm\beta|}$ is a linear combination of monomials of the form $x^{\bm\gamma}y^{\bm\delta}$ with $|\bm\gamma|\le  n(|\bm\beta|-|\bm\delta|)$ and $|\bm\delta|\le |\bm\beta|$. Thus $x^{\bm\alpha}\Gcal(y^{\bm\beta})$ is a linear combination of monomials of the form $x^{\bm\alpha+\bm\gamma}y^{\bm\delta}$ with $|\bm\alpha+\bm\gamma|=|\bm\alpha|+|\bm\gamma|\le n(m-|\bm\beta|)+ n(|\bm\beta|-|\bm\delta|)=n(m-|\bm\delta|)$, and therefore lies in $V_m$. This proves~\eqref{T:augNEW1_pf_2}.

Assume {\ref{T:augNEW13_19}} holds. By taking $\bm\alpha=0$, we see that $\Gcal(y^{\bm\beta})\in V_{|\bm\beta|}$ for all $\bm\beta$, which in view of \eqref{T:augNEW1_pf_2} and {\ref{T:augNEW13_19}} implies that \eqref{T:augNEW1_pf_1} holds. Conversely, assume \eqref{T:augNEW1_pf_1} holds. Since $\Gamma(1,y^{\bm\beta})=0$, it follows that $\Gcal(y^{\bm\beta})\in V_{|\bm\beta|}$ for all $\bm\beta$ and hence, in view of \eqref{T:augNEW1_pf_2}, that $x^{\bm\alpha}\Gcal(y^{\bm\beta})\in V_m$ whenever $|\bm\alpha|\le n(m-|\bm\beta|)$ and $|\bm\beta|\le m$. Another application of \eqref{T:augNEW1_pf_1} then yields $\Gamma(x^{\bm\alpha},y^{\bm\beta})\in V_m$, and we deduce that {\ref{T:augNEW13_19}} holds.

\subsection{Proof of Corollary~\ref{cordimred_19}}
Lemma~\ref{lemlocmart} implies that $Z_t'=(X_t, PY_t)$ is a special semimartingale. We infer that $Z_t'=(X_t, PY_t)$ is an $E\times\R^{e'}$-valued jump diffusion with extended generator $\Gcal'$ of the form~\eqref{eq:GZ_19}. Indeed, this follows from a straightforward modification of Lemma~\ref{L:JD transformNEW_2} applied to the linear map $\varphi(x,y)=(x,Py)$, observing that $P$ does not need to be invertible to back out the state $x$ from $(x,Py)$. It follows by inspection that $Z_t'=(X_t, PY_t)$ satisfies \eqref{DefPolaug1_19} and, by Lemma~\ref{L:JD transformNEW_2} in conjunction with Lemma~\ref{LemPPchar} applied to $(H(X_t),PY_t)$, property {\eqref{T_augNEW1_19_PJD}} with $e$ replaced by $e'$, as claimed.

\subsection{Proof of Theorem~\ref{T:sub}}

The proof of Theorem~\ref{T:sub} builds on the following two lemmas.

\begin{lem} \label{L:Gf loc_bdd}
Let $\Gcal$ be polynomial on $E$. Then $\Gcal f(x)$ is locally bounded on $E$ for every $f\in C^k(\R^d)$ satisfying the growth condition $|f(x)|\le c(1+\|x\|^k)$ on $\R^d$ for some real $c$ and integer $k\ge 2$.
\end{lem}

\begin{proof}
Write
\begin{equation} \label{eq:locbdd_1}
\begin{aligned}
\Gcal f(x) &= \frac12 \tr( a(x)\nabla^2 f(x)) + b(x)^\top \nabla f(x) + \sum_{2\le|\bm\alpha|\le k-1} \frac{\partial^{\bm\alpha}f(x)}{\bm\alpha!}  \int_{\R^d} \xi^{\bm\alpha} \nu(x,d\xi) \\
& \quad + \int_{\R^d} g(x,\xi) \nu(x,d\xi),
\end{aligned}
\end{equation}
where
\begin{equation} \label{eq:locbdd_2}
g(x,\xi) = f(x+\xi) - \sum_{|\bm\alpha|\le k-1} \frac{\partial^{\bm\alpha}f(x)}{\bm\alpha!} \xi^{\bm\alpha}.
\end{equation}
By Lemma~\ref{LemPPchar} and the $C^k$ smoothness of $f(x)$, the first three terms on the right-hand side of~\eqref{eq:locbdd_1} are locally bounded in $x$ on $E$.

We now bound the remaining term in~\eqref{eq:locbdd_1}. For $\|\xi\|>1$, the assumed polynomial bound on $f(x)$ along with the crude inequality $1+\|x+\xi\|^k \le \|\xi\|^k 2^k(1+\|x\|^k)$ and~\eqref{eq:locbdd_2} yield
\[
|g(x,\xi)| \le \|\xi\|^k \left( c 2^k(1+\|x\|^k) + \sum_{|\bm\alpha|\le k-1} \frac{|\partial^{\bm\alpha} f(x)|}{\bm\alpha!} \right), \qquad \|\xi\|>1.
\]
Next, Taylor's theorem with the remainder in integral form yields
\[ g(x,\xi) = \sum_{|\bm\alpha|=k} \frac{k}{\bm\alpha!} \xi^{\bm\alpha} \int_0^1 (1-t)^{k-1} \partial^{\bm\alpha}f(x+t\xi) dt \]
and hence
\[|g(x,\xi)| \le \|\xi\|^k \sum_{|\bm\alpha|=k} \frac{1}{\bm\alpha!} \max_{\|\xi\|\le1}|\partial^{\bm\alpha}f(x+\xi)|, \qquad \|\xi\|\le 1.
\]
Combining these two bounds gives $|g(x,\xi)| \le \|\xi\|^k M(x)$, where $M(x)$ is a continuous function due to the $C^k$ smoothness of $f(x)$. Consequently,
\[
\int_{\R^d}|g(x,\xi)|\nu(x,d\xi) \le M(x) \int_{\R^d}\|\xi\|^k \nu(x,d\xi),
\]
which is locally bounded in $x$ on $E$ by Lemma~\ref{LemPPchar}.
\end{proof}

\begin{lem} \label{L:Gf Gfn}
Let $\Gcal$ be polynomial on $E$. Then for every $f\in C_b^\infty(\R^d)$ there exists a sequence of functions $f_n\in C^\infty_c(\R^d)$ such that $f_n\to f$ and $\Gcal f_n\to\Gcal f$ locally uniformly on $E$.
\end{lem}

\begin{proof}
Define $f_n(x)$ to be $f(x)$ multiplied by a smooth cutoff function that is equal to one on $B_n=\{x\in\R^d\colon \|x\|\le n\}$. Then $f_n\to f$ locally uniformly. For $n>m$ and $x\in B_m$ we have
\begin{align*}
|\Gcal f_n(x) - \Gcal f(x)| &\le  \int_{\R^d} |f_n(x+\xi) - f(x+\xi)| \nu(x,d\xi) \\
&=  \int_{\R^d} |f_n(x+\xi) - f(x+\xi)| \bm1_{\{\|\xi\|>n-m\}}\nu(x,d\xi) \\
&\le \frac{2\|f\|_\infty}{(n-m)^2} \int_{\R^d} \|\xi\|^2 \nu(x,d\xi).
\end{align*}
By Lemma~\ref{LemPPchar} the right-hand side is locally bounded in $x$ on $E$. Hence $\Gcal f_n \to \Gcal f$ uniformly on $E\cap B_m$ for all $m$.
\end{proof}

We now prove Theorem~\ref{T:sub}. We first prove the statement in Remark~\ref{remJM}. Due to the Feller property, for any $f\in C_0(\R^d)$, we know that $\int_E f(y) p_t(x,dy)$ is jointly continuous in $(t,x) \in [0,\infty)\times E$. A monotone class argument now yields the claim for $p_t(x,A)$. For $p_t(x,x+A)$ we observe that
\[ p_t(x,x+A) = \int_E 1_A(y-x) p_t(x,dy). \]
Because $1_A(y-x)$ is jointly measurable in $(x,y)$ the claim follows also for $p_t(x,x+A)$. Hence $\widetilde\Gcal f(x)$ is well defined by \eqref{eq:sub1_Gtilde}--\eqref{sub1:nutilde} for any bounded $C^2$ function $f(x)$ on $\R^d$.

We next claim that
\begin{equation} \label{eqXX1}
\widetilde\Gcal f(x) = b^Z \Gcal f(x) + \int_0^\infty \int_E\left( f(y) - f(x) \right)p_\zeta(x,dy)  \nu^Z(d\zeta)
\end{equation}
for all $f\in \Pol(E)$. Indeed, for any polynomial $f(x)=(1,H(x)^\top) \vec f$ in $\Pol_n(E)$, Theorem~\ref{T:moments} yields
\begin{equation}\label{eqXXXs}
  \int_E\left( f(y) - f(x) \right)p_\zeta(x,dy) =  (1,H(x)^\top) (\e^{\zeta G} - \id) \vec f .
\end{equation}
Hence the right hand side of \eqref{eqXX1} is well defined due to \eqref{eq:nuZ exp moments} for all $x\in E$. Moreover, by
\eqref{eq:sub1_Gtilde}--\eqref{sub1:nutilde} we infer that
\begin{equation}\label{eqXXXu}
 \begin{aligned}
\widetilde\Gcal f(x) &= b^Z \Gcal f(x) + \int_0^\infty \int_E (y - x)^\top \nabla f(x) p_\zeta(x,dy)  \nu^Z(d\zeta) \\
&\quad + \int_0^\infty \int_E \left(f(y)-f(x)-(y-x)^\top \nabla f(x) \right)1_{\{ y\neq x\}} p_\zeta(x,dy) \nu^Z(d\zeta),
\end{aligned}
\end{equation}
which proves \eqref{eqXX1}.

We next claim that the jump-diffusion operator
\begin{equation}\label{claimXX}
  \text{$\widetilde \Gcal$ is polynomial on $E$.}
\end{equation}
First, we have
\[ \int_{\R^d} \|\xi\|^{n}\,\widetilde\nu(x,d\xi) = \int_{\R^d} \|\xi\|^{n}\,b^Z\nu(x,d\xi) + \int_0^\infty \int_{\R^d} \|y-x\|^{n}\, p_\zeta(x,dy)  \nu^Z(d\zeta) <\infty \]
for all $x\in E$ and all $n\ge 2$. Indeed, the first term on the right hand side is finite due to \eqref{DefPol1}. The second term is also finite. This follows from \eqref{eqXXXs} for $f(y)=\|y-x\|^{n}$ and \eqref{eq:nuZ exp moments} for even $n\ge 2$. Second, $\widetilde\Gcal f(x)=0$ on $E$ for any $f\in \Pol(\R^d)$ with $f(x)=0$ on $E$. This follows from \eqref{DefPol2} and \eqref{eqXX1}. Hence $\widetilde\Gcal$ is well-defined on $\Pol(E)$. Finally, the polynomial property of $\Gcal$, \eqref{eqXX1}, and \eqref{eqXXXs} again imply that $\widetilde\Gcal$ maps $\Pol_n(E)$ to itself for each $n\in\N$, which proves \eqref{claimXX}.

Now let $G$ and $\widetilde G$ be the matrix representations of $\Gcal$ and $\widetilde\Gcal$ on $\Pol_n(E)$. The first equation in \eqref{eqGGtilde} then follows from \eqref{eqXX1} and \eqref{eqXXXs}. The right hand side of the second equation equals $M(t)=\E[\e^{Z_t G}]$, which is well-defined by Remark~\ref{remLK}. Due to the L\'evy property of $Z_t$, we have that $M(t+s)=M(t)M(s)$. Hence $M(t)=\exp(t   \dot M(0))$ where $\dot M(0) = \Gcal^Z \e^{zG}|_{z=0} = \widetilde G$. This proves~\eqref{eqGGtilde}.

It remains to verify that $\widetilde X_t$ is a jump-diffusion with respect to $\widetilde\Fcal_t$, and that its extended generator is $\widetilde\Gcal$. By \citet[Theorem~II.2.42]{Jacod/Shiryaev:2003} it suffices to prove that the process
\[
M^f_t = f(\widetilde X_t) - f(\widetilde X_0) - \int_0^t \widetilde\Gcal f(\widetilde X_s)ds
\]
is well-defined and a local martingale for every $f\in C^\infty_b(\R^d)$. We do this in three steps.

First, Phillips' theorem~\cite[Theorem 32.1]{sat_99} shows that $\widetilde p_t(x,dy)$ given in \eqref{deftildept} is a Feller transition kernel and the domain of its generator contains $C^\infty_c(E)$ on which it coincides with the operator
\begin{equation} \label{T:sub:1}
\overline\Gcal f(x) = b^Z \Gcal f(x) + \int_0^\infty \int_E\left( f(y) - f(x) \right)p_\zeta(x,dy)  \nu^Z(d\zeta).
\end{equation}
Here, the integral with respect to $\nu^Z(d\zeta)$ is understood as the Bochner integral of the $C_0(E)$-valued map $\zeta\mapsto u_\zeta$, where $u_\zeta(x)=\int_E\left( f(y) - f(x) \right)p_\zeta(x,dy)$; see \citet[comment after Theorem~32.1]{sat_99}. In particular, when evaluated at a point $x\in E$, this integral coincides with the Lebesgue integral with respect to $\nu^Z(d\zeta)$ of the $\R$-valued function $\zeta\mapsto\int_E\left( f(y) - f(x) \right)p_\zeta(x,dy)$, which is thus well defined and finite. In view of Remark~\ref{rembtilde} therefore \eqref{eqXXXu}, and thus \eqref{eqXX1}, also hold for all $f\in C^\infty_c(E)$. We conclude that
\begin{equation} \label{Gbar=Gtilde}
\text{$\widetilde\Gcal f(x) = \overline\Gcal f(x)$ on $E$ for all $f\in C^\infty_c(E)$.}
\end{equation}

Second, an argument based on \citet[Proposition~III.1.4]{rev_yor_99} shows that $\widetilde X_t$ is a Markov process with respect to its natural filtration $\widetilde\Fcal^0_t=\sigma(\widetilde X_s, s\le t)$ with transition kernel $\widetilde p_t(x,dy)$. Because of the Feller property this also holds with respect to the usual right-continuous augmentation $\widetilde\Fcal_t$ of $\widetilde\Fcal^0_t$, see \citet[Proposition~III.2.10]{rev_yor_99}. Therefore, by \citet[Proposition~VII.1.6]{rev_yor_99} and \eqref{Gbar=Gtilde}, it follows that $M^f_t$ is a martingale for every $f\in C^\infty_c(\R^d)$.

Third, let $f\in C^\infty_b(\R^d)$. In view of \eqref{claimXX} and Lemma~\ref{L:Gf loc_bdd}, $\widetilde\Gcal f(x)$ is locally bounded on $E$, whence the process $M^f_t$ is well-defined. Furthermore, Lemma~\ref{L:Gf Gfn} yields a sequence of functions $f_n\in C^\infty_c(\R^d)$ such that $f_n\to f$ and $\widetilde\Gcal f_n\to\widetilde\Gcal f$ locally uniformly on $E$. Therefore, defining the stopping times $T_m=\inf\{t\ge0\colon \|\widetilde X_t\|\ge m\}$, we have
\[
| M^f_{t\wedge T_m} - M^{f_n}_{t\wedge T_m} | \le  |f(\widetilde X_{t\wedge T_m}) - f_n(\widetilde X_{t\wedge T_m}) - f(\widetilde X_0) + f_n(\widetilde X_0)| + t \max_{x\in E,\,\|x\|\le m} |\widetilde \Gcal f(x) - \widetilde \Gcal f_n(x)|.
\]
The right-hand side is bounded and converges to zero as $n\to\infty$, which yields $M^{f_n}_{t\wedge T_m} \to M^f_{t\wedge T_m}$ in~$L^1$. Since each $M^{f_n}_t$ is a martingale, it follows that $M^f_t$ is a local martingale, as required. This finishes the proof of Theorem~\ref{T:sub}.

\end{appendix}

\bibliographystyle{plainnat}
\bibliography{bibl}

\end{document}